\documentclass[12pt,english]{article}
\usepackage[T1]{fontenc}
\usepackage[latin9]{inputenc}
\usepackage{geometry}
\usepackage{setspace}
\usepackage{color}
\usepackage{babel}
\usepackage{url}
\usepackage{amsmath}
\usepackage{amsthm}
\usepackage{amssymb}
\usepackage{graphicx}
\usepackage[authoryear]{natbib}
\usepackage[unicode=true,
 bookmarks=true,bookmarksnumbered=false,bookmarksopen=false,
 breaklinks=false,pdfborder={0 0 0},pdfborderstyle={},backref=false,colorlinks=true]
 {hyperref}
\hypersetup{pdftitle={Screening p-Hackers: Dissemination
Noise as Bait},
 pdfpagelayout=OneColumn, pdfnewwindow=true, pdfstartview=XYZ, plainpages=false, urlcolor=[rgb]{0.0430 ,0, 0.5}, linkcolor=[rgb]{0.0430 ,0, 0.5}, citecolor=[rgb]{0.0430 ,0, 0.5}, hypertexnames=false}

\makeatletter

\providecommand{\tabularnewline}{\\}

\theoremstyle{plain}
\newtheorem{lem}{\protect\lemmaname}
\theoremstyle{plain}
\newtheorem{prop}{\protect\propositionname}

\def\redherr{r}

\usepackage{babel}

\def\E{\mathbf{E}}

\def\ep{\varepsilon}

\def\E{\mathbb{E}}

\def\ta{\theta}

\def\da{\delta}

\makeatother

\providecommand{\lemmaname}{Lemma}
\providecommand{\propositionname}{Proposition}

\begin{document}
\global\long\def\E{\mathbb{E}}%
 
\global\long\def\da{\mathbb{\delta}}%
 
\global\long\def\ta{\theta}%
 
\global\long\def\ep{\varepsilon}%

\title{Screening $p$-Hackers: Dissemination
Noise as Bait\thanks{This research was made possible through the support of the Linde Institute at Caltech. Echenique also thanks the NSF's support through the grants SES-1558757 and CNS-1518941. We are grateful for comments from Sylvain Chassang, Eva Jin, Albert Ma, Pascal Michaillat, Marco Ottaviani, Nathan Yoder, and the audiences at the University of Pennsylvania, Caltech, Columbia University, Universidad de la Rep\'{u}blica, Boston University,  ASSA 2022, and ACM EC '22. Alfonso Maselli provided excellent research assistance.}}
\author{Federico Echenique\thanks{University of California, Berkeley. Email: \texttt{\protect\href{mailto:fede\%40econ.berkeley.edu}{fede@econ.berkeley.edu}}}
\and Kevin He\thanks{University of Pennsylvania. Email: \texttt{\protect\href{mailto:hesichao\%40gmail.com}{hesichao@gmail.com}}}}
\date{{\normalsize{}}%
\begin{tabular}{rl}
First version: & March 16, 2021\tabularnewline
This version: & 
March 31, 2024 
\end{tabular}
}

\maketitle

\begin{abstract}
{\normalsize{}\thispagestyle{empty}
\setcounter{page}{0}}{\normalsize\par}

We show that adding noise before publishing data effectively screens $p$-hacked findings: spurious explanations produced by fitting many statistical models (data mining). Noise creates ``baits'' that affect two types of researchers differently. Uninformed $p$-hackers, who are fully ignorant of the true mechanism and engage in data mining, often fall for baits. Informed researchers, who start with an ex-ante hypothesis, are minimally affected. We show that as the number of observations grows large, dissemination noise asymptotically achieves optimal screening. In a tractable special case where the informed researchers' theory can identify the true causal mechanism with very little data, we characterize the optimal level of dissemination noise and highlight the relevant trade-offs.  Dissemination noise is a tool that statistical agencies currently use to protect privacy. We argue this existing practice can be repurposed to screen $p$-hackers and thus improve research credibility.

\medskip{}

\end{abstract}

\newpage

\section{Introduction}

In the past 15 years, academics have become increasingly concerned with the harms of \textit{$p$-hacking}: researchers' degrees
of freedom that lead to spurious empirical findings. For the observational
studies that are common in economics and other social sciences, $p$-hacking
often takes the form of multiple testing: attempting  many regression specifications on
the same data with different explanatory variables, without an ex-ante hypothesis, and then selectively reporting
the results that appear statistically significant. Such $p$-hacked
results can lead to misguided and harmful policies, based on a mistaken understanding of the causal relationships between different variables. Recent developments
in data and technology have also made $p$-hacking easier:  today's
rich datasets often contain a large  number of covariates that can
be potentially correlated with a given outcome of interest, while
 powerful computers enable faster and easier specification-searching.

In this paper,  we propose to use \textit{dissemination noise} to address and mitigate the negative effects of $p$-hacking. Dissemination noise is pure white noise that is intentionally added to raw data before the dataset is made public. Statistical agencies, such as the US Census Bureau, already use dissemination noise to protect respondents' privacy.  Our paper suggests that dissemination noise may be repurposed to screen out $p$-hackers. Noise can limit the ability of $p$-hackers to ``game'' standards of evidence by presenting spurious but statistically significant results as genuine causal mechanisms.  We show that the right amount of noise can serve as an impediment to $p$-hacking, while minimally impacting honest researchers who use data to test an ex-ante hypothesis.

\subsection{$p$-Hacking}
Spurious results in many areas of science have been ascribed to the ability of researchers to, consciously or not, vary procedures and models to achieve statistically significant results. The reproducibility crisis in psychology has been blamed to a large extent on $p$-hacking \citep*{simmons2011false, simonsohn2014p, open2015estimating}.  \citet{Camerer1433} evaluate experiments in economics and   find a  significant number of experiments that do not replicate.\footnote{See also \cite{camerer2018evaluating} and \cite{altmejd2019predicting}. \cite*{imai2017common} find evidence against $p$-hacking in experimental economics.}

Most empirical work in economics and other social sciences are observational studies that use existing field data,  not experiments that produce new data. Observational studies lead to a different sort of challenge for research credibility, where  $p$-hacking  stems mostly from discretion in   choosing explanatory variables and econometric specifications. In experimental work, one remedy for $p$-hacking is pre-registration: researchers must describe their methods and procedures before data is collected. But this solution is not applicable for observational studies, because researchers may have already accessed the public dataset before pre-registering.

\subsection{Dissemination Noise}

Dissemination noise is currently used by major statistical agencies to protect people's privacy.  The US Census Bureau, for instance,  only  disseminates a noisy version of the data from the 2020 Census. The practice is not new. Previously, the Bureau has released a tool called ``On the map'' whose underlying data was infused with noise. Even earlier technologies for preserving respondent confidentiality like swapping data and imputing data can also be interpreted as noisy data releases. The contribution of our paper is to propose a new use for dissemination noise.

\subsection{Setup and Key Results}

We consider a society that wants to learn the true cause behind an outcome variable. Researchers differ in their expertise: some are  \textit{mavens} whose domain knowledge  narrows down the true cause to a small set of candidates, and others are \textit{hackers} with no prior information about the true cause. Researchers derive utility both from reporting the true cause and from  influencing policy decisions. So uninformed hackers have an incentive to ``game'' the system by using the data to fish for a covariate that would convince the policymaker. 

We show that dissemination noise can help screen  researcher expertise by  introducing spurious correlations that can be proven to be spurious. These noise-induced correlations act like \textit{baits} for $p$-hackers.  But at the same time, they also make the data less useful for the informed mavens who use the data to test a specific ex-ante hypothesis. 

We explore this trade-off in our model. We show that as the number of observations grows large, dissemination noise asymptotically achieves optimal screening. In a tractable special case where the informed researchers' theory can identify the true cause with very little data, we characterize the optimal level of dissemination noise and derive comparative statics. The key intuition is that \emph{a small amount of noise hurts hackers
more than mavens}. All researchers act strategically to maximize
their expected payoffs, but their optimal behavior differ. Mavens only entertain  a small number of hypotheses,
so a small amount of noise does not greatly affect their
chances of detecting the truth. Hackers, by contrast, rationally try
out a very large number of covariates because they have
no private information about the true cause.
The hackers' data mining amplifies the effect of even a small amount
of noise, making them more likely to fall for a bait and get screened
out. So, adding noise grants an extra informational advantage to the mavens, whose prior knowledge pinpoints a few candidate covariates.   The hackers get screened out precisely because they (rationally) $p$-hack out of complete ignorance about the true cause. 

We focus on a setting where the types of researchers primarily differ in terms of their expertise, not their incentives or biases. While our main results do allow the mavens and the hackers to assign different weights to correctly reporting the true cause versus influencing policy-making, our results are mainly driven by the fact that only the mavens have private information about the true cause. In terms of the classification of different kinds of $p$-hacking practices (see, for example, \cite{maccoun2017blind}), we focus on the problem of deterring  \emph{capitalization on chance}, where the researcher has no preconceived story but fishes around for anything that appears statistically significant in the data. We are not studying  \emph{confirmation bias}, where  a researcher with a preconceived story  looks for evidence that supports the story while discarding or downplaying evidence to the contrary.

We use a stylized model to represent researchers analyzing existing observational data for associations. Our intention is to explore a new channel for screening researcher expertise in a simple and tractable setup. Of course, the practical usefulness of dissemination noise will need to be  evaluated in more specific and realistic domains. Also,  our focus is on simple correlational studies that use existing data: other research designs such as experiments that acquire new data or sophisticated econometric methods that exploit special structure of the data to credibly infer causation are outside of the scope of this work. 

\subsection{Alternative Solutions to $p$-Hacking}
As already mentioned, the most common proposal to remedy  $p$-hacking is pre-registration. While it is a very good idea in many scientific areas,  it is of limited use for observational studies, which are ubiquitous in the social sciences. Not only does it preclude useful exploratory work, it is also impossible to audit or enforce because publicly available data  can be privately accessed by a researcher before pre-registration.

A second solution is to change statistical conventions  and make $p$-hacking more difficult. An extreme example is banning the use of statistical inference altogether \citep{woolston2015psychology}. A less drastic idea is contained in  \citet{benjamin2018redefine}, which proposes to lower the $p$-value threshold for statistical significance by an order of magnitude --- from  5\% to 0.5\%. Of course this makes $p$-hacking harder, but a $p$-hacker armed with a sufficiently ``wide'' dataset and cheap enough computation power can discover spurious correlations that satisfy any significance threshold. We address this idea within our model and argue that our proposed use of dissemination noise is largely complementary to requiring more demanding statistical significance. 

An idea related to our proposal is simply to reserve data for out-of-sample testing. Typically, the observations are partitioned into two portions. One portion is released publicly, and the rest is a ``hold-out'' dataset reserved for out-of-sample testing. We instead focus on a model of noise where each observation of each covariate is independently perturbed, which more closely resembles the kind of dissemination noise currently in use for privacy purposes. Our central message is that the current implementation of noise can be repurposed to screen out $p$-hacking. In addition, the kind of dissemination noise we study may be more applicable for datasets where the observations are not generated from an i.i.d. process, and thus it is less reasonable to designate some observations as  a hold-out dataset (e.g., observations are different days in a time series).

The out-of-sample approach is the focus of \citet{Dwork636}. We differ in that we consider a world with two kinds of researchers and the dissemination noise here serves a screening role  to separate the two types who act strategically to maximize their expected payoffs. 

\subsection{Related Literature}

In economics, there is a recent strand of literature that seeks to understand the incentives and trade-offs behind $p$-hacking. 
\cite*{henry2009strategic,felgenhauer2014strategic,felgenhauer2017bayesian,di2017strategic,henry2019research,mccloskey2022critical} all study different games between a researcher (an agent) and a receiver (a principal). The agent has access to some  $p$-hacking technology, which takes various forms such as repeatedly taking private samples and then selectively reporting a subset of favorable results to the principal, or  sampling publicly but strategically stopping when  ahead. These papers seek to  better understand the equilibrium interaction between $p$-hacking agents and their principals, and study how such interactions are affected by variations in the hacking technology. 

This literature differs from our work in two ways. First, they consider the case where hacking is costly. On this dimension, these papers about $p$-hacking are related to the broader literature on ``gaming,'' where agents can undertake costly effort to improve an observable signal (here, the $p$-value) beyond its natural level (e.g.,  \cite{frankel2019muddled}). We instead consider hackers who incur zero cost from $p$-hacking, motivated by our focus on researchers who data mine an existing dataset (which is essentially free with powerful computers). Absent any interventions, equilibria  with zero hacking or gaming cost would be uninteresting. Our focus is instead on a specific intervention, dissemination noise, that can help screen out the $p$-hackers even though they face no hacking costs.  The second difference is that these papers do not consider the problem of expertise screening. In our world, the principal's main problem is to provide sufficiently informative data to agents who have expertise  while distorting the data enough to mislead another type of agent who try to make up for their lack of expertise with $p$-hacking.

\cite*{di2017persuasion} also study a game between a $p$-hacker and a principal, but give the agent some private information and the ability to select an area to do research in. This is a mechanism for hacking that is outside of the scope of our paper.

\section{\label{sec:model}Model and Asymptotically Optimal Screening Using Dissemination Noise}

\subsection{The Baseline Model}

We propose a model that captures the essence of how dissemination noise allows for expertise screening in an environment where non-expert agents can $p$-hack,  while keeping the model tractable enough to allow for analytic solutions.

\subsubsection{The Raw Dataset}

Consider an environment where each unit of observation is associated with an outcome  $Y$ and a set $A$ of potential causal covariates $(X^{a})_{a\in A}$. The outcome and each covariate is binary. Suppose the dataset is ``wide,'' so the set of potential causes for the outcome is large relative to the number of observations. In fact, we assume  a continuum of covariates; so $A=[0,1]$. For instance, the covariates may indicate the presence or absence of different SNPs in a person's genetic sequence, while the outcome refers to the presence or absence of a certain disease. 

There is one covariate $a^{*} \in A$, the \emph{true cause}, with $\mathbb{P}[X^{a^{*}}=Y]=\psi$ for some $\psi\in (1/2,1]$. So the true cause is positively correlated with the outcome, but it may not be perfectly correlated.  For instance,  $a^{*}$ is the one SNP that causes the disease in question. There is also a \emph{red herring}  covariate $a^{\redherr}\in A$ that is independent of $Y$.  The red herring 
represents a theoretically plausible mechanism for the outcome $Y$
that can only be disproved with data. For instance,  $a^{\redherr}$ might be a SNP that seems as likely to cause the disease as $a^*$ based on a biological theory about the roles of different SNPs. 

Nature draws the true cause $a^{*}$ and the red herring $a^{\redherr}$, independently and uniformly from $A$. Then Nature generates the raw dataset $(Y_n, (X_n^a)_{a \in [0,1]})$ for observations $1 \le n \le N$. First, the values of the true cause  in the $N$ observations $(X_{n}^{a^{*}})_{1\le n\le N}$  are generated independently,
each equally likely to be 0 or 1. Then, each $Y_{n}$ is generated
to match $X_{n}^{a^{*}}$ with probability $\psi$  for $1\le n\le N$, independently
across $n.$ Finally, covariates $X_{n}^{a}$ for $a\ne a^{*}$, $1\le n\le N$
are generated, each equally likely to be 0 or 1, independent of each
other and of all other random variables. (So there is a continuum
of independent Bernoulli random variables.) Equivalently,
once $a^{*}$ and $a^{\redherr}$ are drawn, we have fixed
a joint distribution between $Y$ and the covariates $(X^{a})_{a\in A}$,
and the raw dataset consists of $N$ independent draws from this joint distribution. For instance, this may represent a dataset that shows the complete genetic sequences of $N$ individuals and whether each person suffers from the disease. 

\subsubsection{Players and Their Incentives}

There are three players in the model: a principal, an agent, and a
policymaker. The \emph{principal} owns the raw dataset, but lacks
the ability to analyze the data and cannot influence policy-making
norms. The principal disseminates a noisy version of the dataset,
which we describe below. The \emph{agent} uses the disseminated data
to propose a covariate, $\hat{a}$. Finally, a \emph{policymaker}
evaluates the agent's proposal on the raw dataset using an exogenous
test. We think of the agent as proposing an intervention: if this
proposal passes, the policymaker will implement a policy that changes
$X^{\hat{a}}$ in order to affect the value of $Y$. In the background, we implicitly assume that the principal grants the policymaker access to the raw data to conduct the test. (In Section \ref{sec:dynamicproblem}, we model the principal periodically noisy versions of the raw data for these tests. Such data releases will diminish the principal's ability to screen out $p$-hackers in the future.)

The policymaker's role is mechanical, and restricted to deciding whether
the agent's proposal passes an exogenous test. We say that the proposal
$a$ \emph{passes} if the covariate $X^{a}$ equals the outcome $Y$
in $M=\left\lfloor \gamma\cdot N\right\rfloor $ out of $N$ observations,
and that it \emph{fails} otherwise. The parameter $\gamma$ is an
exogenous passing threshold with $1/2<\gamma<\psi$. The
policymaker will adopt a policy proposal if and only if it passes
the test on the raw data. Passing the test does not require $a$ to be the true cause of $Y$, for we could have some
covariate $a \ne a^*$ where $Y_{n}=X_{n}^{a}$ for at least $M$ observations by random
chance.\footnote{There is no reward in our model for disproving a hypothesis.}

The agent is either a maven (with probability $1-h$) or a hacker (with probability $h$). Mavens and hackers differ in their expertise. A maven knows that the true cause is either $a^{*}$ or $a^{\redherr}$, and assigns them equal probabilities, but a hacker is ignorant about the realizations of $a^{*}$ and $a^{\redherr}$. The idea is that a maven uses domain knowledge (e.g., biological theory about the disease $Y$)
to narrow down the true cause to the set $\{a^{*},a^{\redherr}\}$. A hacker, by contrast, is completely uninformed about the mechanism causing $Y$. 

The agent's payoffs reflect both a desire for reporting the true cause and a desire for policy impact. If a type $\theta$ agent proposes $a$ when the true cause is $a^{*}$, then his payoff is 
\[
w_{\theta}\cdot\boldsymbol{1}_{\{a=a^{*}\}}+(1-w_{\theta})\cdot\boldsymbol{1}_{\{\text{at least }\left\lfloor \gamma \cdot N\right\rfloor \text{ observations }n\text{ have }Y_{n}=X_{n}^{a}\}}.
\]
Here we interpret $\boldsymbol{1}_{\{a=a^{*}\}}$ as the effect of
proposing $a$ on the agent's long-run reputation when the true cause
$a^{*}$ of the outcome $Y$ eventually becomes known. The other summand 
models the agent's gain from proposing a policy that passes the policymaker's test and gets
implemented. The relative weight $w_{\theta}\in[0,1]$ on
these two components may differ for the two types of agent.  Our main results
in this section are valid for any  values of $w_{\text{maven}}$ and $w_{\text{hacker}}$ in $[0,1]$, but some later results in Section \ref{sec:special} will require restrictions on $w_{\text{maven}}$.

The principal obtains a payoff of $1$ if a true cause passes, a payoff
of $-1$ if any other $a\ne a^{*}$ passes, and a payoff of 0 if the
agent's proposal is rejected. The principal's payoff reflects an objective of maximizing the
positive policy impact of the research done on their data.

\subsubsection{Dissemination Noise}

The principal releases a noisy dataset $\mathcal{D}(q)$ by perturbing
the raw data. Specifically, they choose a \emph{level of noise} $q\in[0,1/2]$ and
every binary realization of each covariate is flipped independently with
probability $q$. So the noisy dataset $\mathcal{D}(q)$ is $(Y_{n},(\hat{X}_{n}^{a})_{a\in A})$,
where $\hat{X}_{n}^{a}=X_{n}^{a}$ with probability $1-q$, and $\hat{X}_{n}^{a}=1-X_{n}^{a}$
with probability $q$. The principal's choice of $q$ is common knowledge. A covariate $a$ that matches the outcome in at least $M$ observations in the noisy dataset but would not pass the policymaker's test --- that
is $\hat{X}_{n}^{a}=Y$ for  at least $M$  observations  but $X_{n}^{a}=Y$
for fewer than $M$ observations --- is called a \emph{bait}. 

The form of noise in our model is  motivated by the  dissemination noise currently
in use by statistical agencies, like the US Census Bureau. One could imagine other ways of generating a ``noisy'' dataset, such as
selecting a random subset of the observations and making them fully uninformative, which corresponds to reserving the selected observations
as a ``hold-out'' dataset for out-of-sample testing. Our analysis explores the possibility of repurposing the existing practice of adding dissemination noise, which more closely resembles perturbing each data entry independently than withholding
some rows of the dataset altogether.\footnote{For instance, the Bureau publishes
the annual Statistics of U.S. Businesses that contains payroll and employee data of small
U.S. businesses. Statisticians at the Bureau say that separately adding noise to each business
establishment's survey response provides ``an alternative to cell
suppression that would allow us to publish more data and to fulfill
more requests for special tabulations'' \citep*{evans1998using}. The dataset has
been released with this form of  noise since 2007 \citep{SUSB}.
For the 2020 Census data, the Bureau will add  noise  through the new differentially private TopDown Algorithm that replaces the previous methods of data
suppression and data swapping \citep{DP_presentation_slides}.}

\subsubsection{Remarks about the Model}
We comment on our assumptions regarding the agents and the data in our model.

First, our model features very powerful $p$-hackers. A fraction $h$ of researchers are totally ignorant about the true cause, 
but they are incentivized to game the system by fishing for some covariate that plausibly explains the outcome variable and passes
the policymaker's test. This kind of $p$-hacking by multiple hypothesis testing is made easy by the fact that hackers have a continuum of
covariates to search over and incur no cost from data mining. Our assumptions 
represent today's ``wide'' datasets and powerful computers that enable
ever easier $p$-hacking. Our analysis suggests that dissemination noise can improve social welfare, even in settings where $p$-hacking is costless.

Second, the principal is an entity that wishes to maximize the positive social
impact of the research done using their data, but has limited power in influencing
the institutional conventions surrounding how research results are evaluated and implemented into
policies. In the model, the principal cannot change the policymaker's test. Examples include private firms like 23andMe that possess a unique dataset but have little
say in government policy-making, and agencies like the US Census Bureau that are charged with data collection
and data stewardship but do not directly evaluate research conclusions. Such organizations already
introduce intentional noise in the data they release for the purpose of protecting individual privacy, so they may
be willing to use the same tool to improve the quality of policy interventions guided by studies done on their data. In line with this interpretation of the principal, they cannot influence the research process or the policymaker's decision, except through changing the  quality of the disseminated data. In particular, 
the principal cannot impose a cost on the agent to submit a proposal to the policymaker, write a contract to punish an agent who proposes a misguided policy, or change the protocols surrounding how proposals get tested and turned into policies.

Third, the dataset in our model contains just one outcome variable,
but in reality a typical dataset (e.g., the US Census data) contains
many outcome variables and can be used to address many different questions.
We can extend our model to allow for a countably infinite number of
outcome variables $Y^{1},Y^{2},\ldots$, with each outcome associated
with an independently drawn true cause and red herring. After the
principal releases a noisy version of the data, one random outcome
becomes  relevant and the agent proposes a model for this
specific outcome. Our analysis remains unchanged in this world. This more realistic setting provides a foundation for the principal
not being able to screen the agent types by eliciting their private
information about the true cause without giving them any data. Who is a maven depends on
the research question and the outcome variable being studied, and it is infeasible to test a researcher's domain expertise
with respect to every conceivable future research question. 

Fourth, the policymaker's exogenous test only evaluates how well the  agent's proposal explains the raw dataset, and does not provide the agent any other way to communicate his domain expertise. Such a convention may arise if domain expertise is complex and difficult to  convey credibly: for instance, an uninformed hacker who has found a strong association in the data  can always invent a plausible-sounding story to justify why a certain covariate causes the outcome. We also assume that the policymaker's test is mechanically set and does not adjust to the presence of $p$-hackers. This represents a short-run stasis in the science advocacy process or publication norms --- for instance, while we know how to deal with multiple hypotheses testing, a vast majority of academic journals today still treat $p < 0.05$ as a canonical cutoff for statistical significance. Our analysis suggests that dissemination noise can help screen out misguided policies in the short run, when the principal must take as given a policymaking environment that has not adapted to the possibility of $p$-hacking. 

To conclude, our basic model is meant to isolate the tradeoff between the cost imposed by noise on honest researchers and the benefit of screening $p$-hackers. A discussion of how the results are affected by relaxing our assumptions is in Section~\ref{sec:conclusion}.

\subsection{Screening Using Noise}

We first derive the optimal behavior of
the hacker and the maven. 

\begin{lem}
\label{lem:agent_behavior_new}For any $q\in[0,1/2)$, it is optimal for
the hacker to propose any $a\in A$ that satisfies $\hat{X}_{n}^{a}=Y_{n}$
for every $1\le n\le N$. It is optimal for the maven to either propose
 $a \in \{a^{*}, a^{\redherr}\}$ that maximizes the
number of observations $n$ for which $\hat{X}_{n}^{a}=Y_{n}$ (and
randomize uniformly between the two covariates if there is a tie) or to propose any $a\in A$ that satisfies $\hat{X}_{n}^{a}=Y_{n}$
for every $1\le n\le N$. 
\end{lem}

Given the policymaker's exogenous test, hackers find it optimal to ``maximally $p$-hack.'' Depending on the relative weight $w_\text{maven}$ that mavens put on reporting the true cause, they will either use the noisy data to decide between their two true-cause candidates or engage in $p$-hacking. If the principal releases data without noise, then hackers will propose a covariate that is perfectly correlated with $Y$ in the raw data. This covariate passes the policymaker's test, but it leads to a misguided policy with probability 1. The payoff to the principal from releasing the data without noise is therefore no larger than $1-2h$. 

In fact, the principal cannot hope for an expected payoff higher
than $1 - h$. This first-best benchmark corresponds to the
policymaker always implementing the correct policy when
the agent is a maven and not implementing any policy when
the agent is a hacker. We show that with an appropriate level of dissemination noise, the principal's expected payoff  approaches this first-best benchmark as the number of observations grows large. 

\begin{prop}\label{prop:asymptotic_general} 
    For every $q$ with $1-\gamma<q<1/2$, the principal's payoff from using dissemination noise $q$ converges to $1-h$ as $N \to \infty$.
\end{prop}

That is to say, dissemination noise is asymptotically optimal among all mechanisms for screening
the two agent types, including mechanisms that take on more complex forms that we
have not considered in our analysis.

The intuition is that noise does not prevent the
agent from finding a  policy that passes the policymaker's test if his private information narrows down the true covariate to a small handful of candidates. But if the agent
has a very large set of candidate covariates, then there is a good chance
that the noise turns several covariates from this large set into baits. For example, if $N=100$, $\psi = 0.95$, $\gamma = 0.9$,  and $q=0.15,$ a covariate that perfectly correlates with $Y$ in the noisy dataset has a 90\%  probability of being a bait. In the same environment, a maven who restricts attention to only two covariates  ($a^*$ and $a^{\redherr}$) and proposes the covariate that correlates more with the outcome only fails the policymaker's test about 1\% of the time. (As $N$ grows for a fixed value of $q$ in the range given by Proposition \ref{prop:asymptotic_general}, the probability of a maven  proposing a covariate other than  $a^*$ or $a^{\redherr}$ under his optimal strategy converges to zero). Hackers fall for baits at a higher rate than mavens because
they engage in $p$-hacking and try out multiple hypotheses. Yet $p$-hacking is the hackers' best
response, even though they know that the dataset contains baits. 

While Proposition \ref{prop:asymptotic_general} applies asymptotically, the next result gives a lower bound on the number of observations such that  a given level of  noise is better than not adding any noise. 

\begin{prop}\label{prop:N_bound} 
For every $q$ with $1-\gamma<q<1/2,$ the principal gets higher expected
payoff with $q$ level of noise than with zero noise whenever 
\begin{align*}
N\ge\max & \left\{ \frac{-\ln(h/8)}{2(q+\gamma - 1)^{2}},\frac{-2\ln(h/32)}{(\psi(1-q)+(1-\psi)q-0.5)^{2}},\right.\\
 & \left.\frac{-\ln(h/16)}{2(\psi-\gamma)^{2}}\right\} .
\end{align*}

\end{prop}

For example, when $\psi = 0.95$, $\gamma = 0.9$, $h=0.1$, this result says $q=0.15$ is better than $q=0$ whenever $N \ge 1016$. 

\section{Optimal Dissemination Noise in a Special Case}\label{sec:special}

We now turn to a tractable special case where we can characterize the optimal level of noise with any finite number $N$ of observations. We make two modifications relative to the baseline model discussed before. 

First, we suppose the environment is such that the maven's theory only requires a minimal amount of data to distinguish $a^*$ from $a^\redherr$. Specifically, suppose we always have  $Y_n = X^{a^*}_n$ and  $X^{a^\redherr}_n = 1 - X^{a^*}_n$ for every observation $n$. Unlike in the baseline model, the true cause is now  perfectly correlated with the outcome $Y$ and  perfectly negatively correlated with the red herring $X^{a^\redherr}$.  We think of $X^{a^*}$ as the causal covariate that determines the values of both $X^{a^\redherr}$ and $Y$. As before, the principal gets 1 if a proposal targeting $a^*$ passes, -1 if any other proposal passes, and 0 if the proposal is rejected. Note that even though $X^{a^\redherr}$ is perfectly negatively correlated with the outcome, it does not cause the outcome. So a policy intervention that changes $X^{a^\redherr}$ is as ineffective at changing the outcome as a policy targeting any other covariate $a \ne a^*$. 

Second, we suppose the policymaker uses the most stringent test. The proposal $a$ passes if and only if $Y_{n}=X_{n}^{a}$
for all $1\le n\le N$. (The principal can only do worse if the policymaker uses a more lenient test, as we will later show.)

As before, suppose agents maximize a weighted sum between reporting the true cause and passing the policymaker's test. Given the form of the test, the type $\theta$ agent's utility from proposing $a$ when the true cause is $a^*$ is: 
\[
w_{\theta}\cdot\boldsymbol{1}_{\{a=a^{*}\}}+(1-w_{\theta})\cdot\boldsymbol{1}_{\{Y_{n}=X_{n}^{a} \text{ for every } 1\leq n\leq N\}}.
\]
We suppose $0 \le w_{\text{hacker}} \le 1$  and $1/2 <  w_{\text{maven}} \le 1$.
\begin{lem}
\label{lem:agent_behavior}For any $q\in[0,1/2]$, it is optimal for
the hacker to propose any $a\in A$ that satisfies $\hat{X}_{n}^{a}=Y_{n}$
for every $1\le n\le N$, and it is optimal for the maven to propose
 $a \in \{a^{*}, a^{\redherr}\}$ that maximizes the
number of observations $n$ for which $\hat{X}_{n}^{a}=Y_{n}$ (and
randomize uniformly between the two covariates if there is a tie). 
\end{lem}
When the agents follow the optimal behavior described in Lemma~\ref{lem:agent_behavior}, the principal's expected
utility from choosing noise level $q$ is $-hV_{\text{hacker}}(q)+(1-h)V_{\text{maven}}(q)$,
where $V_{\theta}(q)$ is the probability that an agent of type $\theta$'s proposal
passes the policymaker's test in the raw data, when the noise level is $q$. The next result formalizes the core idea that a small
amount of noise harms the hackers more than the mavens. 
\begin{lem}
\label{lem:lemmarginalutility} 
$V_{\text{maven}}^{'}(q)=-\binom{2N-1}{N}Nq^{N-1}(1-q)^{N-1}$
and $V_{\text{hacker}}^{'}(q)=-N(1-q)^{N-1}$. In particular, $V_{\text{maven}}^{'}(0)=0$
while $V_{\text{hacker}}^{'}(0)=-N$. 
\end{lem}

We can show that the principal's overall objective $-hV_{\text{hacker}}(q)+(1-h)V_{\text{maven}}(q)$
is strictly concave, and therefore the first-order condition characterizes
the optimal $q$, provided the solution is interior: 
\begin{prop}
\label{prop:mainstatic} If $\frac{h}{1-h}\le\binom{2N-1}{N}(1/2)^{N-1}$
then the optimal noise level is 
$q^{*}=\left(\frac{h}{1-h}\frac{1}{\binom{2N-1}{N}}\right)^{1/(N-1)}.$
More noise is optimal when there are more hackers and less is optimal
when there are more observations. If $\frac{h}{1-h}\ge\binom{2N-1}{N}(1/2)^{N-1}$
then the optimal noise level is $q^{*}=1/2$. 
\end{prop}
Proposition~\ref{prop:mainstatic} gives the optimal dissemination noise in closed form.
With more hackers, screening out their misguided policies becomes
more important, so the optimal noise level increases. With more observations,
the same level of noise can create more baits, so the principal can
dial back the noise to provide more accurate data to help the mavens.

\subsection{Dissemination Noise and $p$-Value Thresholds}\label{sec:pvaluethresholds}

Now suppose the principal can choose both
the level of noise $q\in[0,1/2]$ and a passing threshold $\underline{N}\in\{1,\ldots,N\}$
for the test, so that a proposal passes whenever $X_{n}^{a}=Y_{n}$
for at least $\underline{N}$ out of the $N$ observations. 
\begin{prop}
\label{prop:optimal_threshold} When the principal can optimize over
both the passing threshold and the noise level, the optimal threshold
is $\underline{N}=N$, and the optimal noise level is the same as in
Proposition \ref{prop:mainstatic}.
\end{prop}

We can interpret this result to say that stringent $p$-value thresholds
and dissemination noise are \emph{complementary tools} for screening out
$p$-hackers and misguided policies. Think of different passing
thresholds as different $p$-value thresholds, with the threshold
$N=\underline{N}$ as the most stringent $p$-value criterion that
one could impose in this environment. \citet{benjamin2018redefine}'s
article about lowering the ``statistical significance'' $p$-value
threshold for new findings includes the following discussion: 
\begin{quotation}
``The proposal does not address multiple-hypothesis testing, P-hacking,
{[}\ldots{]} Reducing the P value threshold complements --- but does
not substitute for --- solutions to these other problems.'' 
\end{quotation}
Our result formalizes the sense in which reducing $p$-value thresholds
complements dissemination noise in improving social welfare from research.

\subsection{Further Extensions\label{sec:discussionandextensions}}

We explore several relaxations of our simplifying assumptions for the special case studied in this section. First, we consider non-i.i.d. observations, such as those in  time-series data, or data from social networks.  Second, we look at a model in which the maven can face a red herring that is harder to disprove with data than we have assumed so far. Third, we relax the assumption of a continuum of potential covariates. Finally, we relax the assumption that there exists a correct explanation for the outcome variable in the data.

\subsubsection{Non-i.i.d.\ Observations \label{sec:noniidobservations}}

In the model from Section \ref{sec:special}, for each $a\in A$ the raw data contains $N$
i.i.d.\ observations of the $a$ covariate $X^{a}$. This gives a vector
$X^{a}\in\{0,1\}^{N}$ with independent and identically-distributed
components $X_{n}^{a}$. The i.i.d.\ assumption rules out certain applications
where there is natural dependence between different observations of
the same covariate, such as data from social networks, panel data,
or time-series data. We now relax this assumption. For each policy
$a$ there is associated a covariate $X^{a}\in\{0,1\}^{N},$ but the
ex-ante distribution of $X^{a}$ is given by an arbitrary, full-support
$\mu\in\Delta(\{0,1\}^{N})$. (Full support means that
$\mu(x)>0$ for every $x\in\{0,1\}^{N}.$)

The model is otherwise the same as the one from Section \ref{sec:special}. In particular, the true cause and the red herring covariates still
exhibit perfect correlation and perfect negative correlation with
the outcome variable, viewed as random vectors in $\{0,1\}^{N}$.
To be concrete, Nature first generates $a^{*}$ and $a^{\redherr}$ independently and uniformly from $A$. Then, Nature   generates the outcome variable as a vector, $Y\sim\mu$.
Then Nature sets  $X^{a^{*}}=Y$ and $X^{a^{\redherr}}=1-Y$. Finally, for
each $a\in A\backslash\{a^{*},a^{\redherr}\}$, Nature draws the vector $X^{a}\sim\mu$
independently (and independently of $Y$).

When the principal prepares the noisy dataset $\mathcal{D}(q)$, noise is still added to each observation of each covariate independently with probability $q$.  We first show that the hacker's
payoff-maximizing strategy is still to propose a covariate $a$ with
$Y_{n}=\hat{X}_{n}^{a}$ for every observation $n$ in the noisy data.
That is, regardless of how the $N$ observations are correlated, there
is nothing more ``clever'' that a hacker could do to increase the
probability of passing the test than to ``maximally $p$-hack''
and propose a covariate that appears perfectly correlated with the
outcome variable in the noisy dataset.
\begin{lem}
\label{lem:hacker_non_iid}For any $y\in\{0,1\}^{N}$, $\mathbb{P}[X^{a}=y\mid\hat{X}^{a}=x]$
is maximized across all $x\in\{0,1\}^{N}$ at $x=y$, for any $0\le q \le 1/2$ and full-support $\mu.$
\end{lem}
Using the hacker's optimal behavior in Lemma \ref{lem:hacker_non_iid},
we can show that a small amount of dissemination noise will differentially
impact the two types' chances of passing the test, thus it improves
the principal's expected payoff as in the case when the observations
are i.i.d.
\begin{prop}
\label{prop:noniid} For any full support $\mu\in\{0,1\}^{N},$ $V_{\text{maven}}^{'}(0)=0$
while $V_{\text{hacker}}^{'}(0)<0$. In particular, there exists $\bar{q}>0$ so that any noise level
$0 < q \le \bar{q}$ is strictly better than $q=0$.
\end{prop}
When the raw dataset consists of correlated observations --- for
example, data on $N$ individuals who influence each other in a social
network or $N$ periods of time series data for a very large number of economic indicators ---
it may be unreasonable for the principal to only release some of the
observations (e.g., only the time series data for even-number years)
and keep the rest of the raw dataset as a secret holdout set to test
the agent's proposal and identify the $p$-hackers. Our procedure
of releasing all of the observations infused with i.i.d.\ dissemination
noise (which resembles the current implementation of privacy-preserving noise) may be more reasonable in such contexts.  Proposition \ref{prop:noniid}
shows our main insight continues to be valid. Even when the observations
have arbitrary correlation, which the hackers may take advantage of
in their data mining, a small amount of dissemination noise still
strictly improves the principal's expected payoff.

\subsubsection{More Misleading Red Herrings \label{sec:generalredherring}}

In the model from Section \ref{sec:special}, we assumed that
the ``red herring'' covariate is  perfectly negatively
correlated with $Y$. This corresponds to an extreme kind of complementarity
between theory and data in learning the true cause, as even a small
amount of data can disprove the theoretically plausible alternative
and identify the truth.

We now consider a situation where the red herring is more misleading, 
and not always easily ruled out by the data. We allow the red herring
to be just like any other covariate in $A$, so that it is simply
uncorrelated with the outcome instead of perfectly negatively correlated with it. So, the only modification relative to the  model from  Section \ref{sec:special} is that $X^{a^{\redherr}}$,  like $X^{a}$ for $a \notin \{a^{*},  a^{\redherr}\}$,  is also independent of $Y$.  

It is
easy to see that the change in how we model the red herring covariate
does not affect the optimal behavior of either the hacker or the maven.
A hacker proposes some covariate $a$ that perfectly correlates with
$Y$ in the noisy dataset. A maven chooses between $a^{*}$ and $a^{\redherr}$
according to how they correlate with $Y$ in the noisy data, randomizing
if there is a tie. When the red herring covariate is independent of
the outcome in the raw dataset, the maven falls for the red herring
with a higher probability for every level of noise. Also, unlike in
the baseline model where the maven gets rejected by the policymaker
if he happens to propose the red herring, here the maven may propose
a misguided policy that passes the test if all $N$ realizations of
$X^{a^{\redherr}}$ perfectly match that of the outcome $Y$ in the
raw dataset.

Our next result implies that a strictly positive amount of dissemination
noise still improves the principal's expected payoffs given ``reasonable''
parameter values.
\begin{prop}
\label{prop:generalredherring} The derivative of the principal's
expected payoff, as a function of the noise level $q$, is $hN-(1-h)N(N+1)2^{-(N+1)}$
when evaluated at $q=0$. This derivative is strictly positive when
$h>\frac{N+1}{2^{N+1}+N+1}$.  In particular, when this condition on $h$ is satisfied, there exists $\bar{q}>0$ so that 
any noise level $0 < q \le \bar{q}$ is strictly better than $q=0$.
\end{prop}
When the red herring covariate is perfectly negatively correlated
with the outcome variable, we found that the optimal level of noise
is always strictly positive. Proposition \ref{prop:generalredherring}
says this result remains true even when the red herring can be more
misleading, provided there are enough hackers relative to the number
of observations in the data. The lowest amount of hackers required
for dissemination noise to be useful converges to 0 at an exponential
rate as $N$ grows. For example, even when there are only $N=10$
observations, the result holds whenever more than 0.53\% of all researchers
are $p$-hackers.

\subsubsection{Finite Number of Covariates \label{sec:kvariables}}

In the model from Section \ref{sec:special}, we imagine there is a continuum of covariates
$a\in A=[0,1]$. This represents an environment with a very ``wide''
dataset, where there are many more candidate explanatory variables and econometric specifications than
observations. But the main idea behind our result remains true if there
is a finite but large number of covariates.

Suppose $A=\{1,2,\ldots,K\}$, so there are $2\le K<\infty$ covariates.
As in the baseline model, a true cause and a red herring are drawn
from the set of all covariates, with all pairs $(a^{*},a^{\redherr})\in A^{2},$
$a^{*}\ne a^{\redherr}$ equally likely. The outcome $Y$ is perfectly
positively correlated with the true cause, so $Y=X^{a^{*}}$. The
other covariates (including the red herring) are independent of the outcome, as in the extension in Section \ref{sec:generalredherring}. Once $(a^{*},a^{\redherr})$ are drawn, we have fixed a joint distribution
among the $K+1$ random variables $(Y,X^{1},...,X^{K})$. The raw
dataset consists of $N$ independent observations drawn from this
joint distribution. The principal releases a noisy version of the
dataset with noise level $q\in[0,1/2]$ as before. 

As the number of covariates $K$ grows, there is more scope for $p$-hacking
to generate misguided policies. This happens for two reasons. First,
holding fixed the policymaker's test and the number of observations,
it is easier for the $p$-hacker to find a covariate that passes the
test when there are more covariates to data mine. Second, the probability
that the $p$-hacker proposes an incorrect covariate also increases with
$K$. When the number of covariates is finite, a $p$-hacker has a
positive probability of stumbling upon the true cause by chance, but
this probability converges to 0 as $K$ goes to infinity. As the statistical
environment becomes more complex and the number of potential models
explodes ($K\to\infty$), not only is the $p$-hacker more likely
to pass the test, but his proposal also leads to a misguided policy
with a higher probability conditional on passing.

In fact, the social harm of a $p$-hacker converges to that of the
baseline model with a continuum of covariates as $K\to\infty.$ As
a result, we can show that a small amount of dissemination noise improves
the principal's payoffs relative to no noise when $K$ is finite but
large, provided the fraction of hackers is not too close to 0.
\begin{prop}
\label{prop:finite_K}Let the number of observations $N$ and the
fraction of hackers $0<h<1$ be fixed, and suppose $h>\frac{N+1}{2^{N+1}+N+1}$.
There exists a noise level $q'>0$ and an integer $\underline{K}$
so that when there are $K$ covariates with $K\ge\underline{K},$
the principal does strictly better with noise level $q'$ than noise
level 0.
\end{prop}

The lower bound on the fraction of hackers in this result is mild and matches the condition from Proposition \ref{prop:generalredherring}. If there are 10 observations and more than 0.53\% of researchers are uninformed, then the principal can improve her expected payoff with a non-zero amount of noise whenever the (finite) dataset contains enough covariates. 

\subsubsection{No True Cause}\label{sec:notruecause}

We turn to a version of our environment where all models can be wrong. Suppose that, with some probability, none of the covariates in the dataset
is the causal mechanism behind the outcome. As in the baseline model,
 Nature draws $a^{*}$ and $a^{\redherr}$ uniformly at random
from $[0,1]$. With probability $0<\beta\le1,$ the covariate $a^{*}$
is the true cause and $X^{a^{*}}$ is perfectly correlated with $Y$.
But with the complementary probability, $a^{*}$ is another red herring
and $X^{a^{*}}$ is perfectly negatively correlated with $Y$ (just
as $X^{a^{\redherr}}$ is). The maven observes $a^{\redherr}$ and $a^{*}$ ---
the maven does not know which is which, and does not know whether
$a^{*}$ is the true cause or another red herring.

The agent can either propose a covariate $a\in A$, or report $\varnothing$
indicating that none of the covariates is the true cause. If the agent
proposes a covariate, the policymaker implements it if and only if it passes the policymakers' test 
(that is, if it is perfectly correlated with $Y$ in the original dataset). The
principal gets 1 if the true cause is implemented, 0 if the proposal
is rejected, and -1 if any other covariate is implemented: so when
the data does not contain the true cause, the principal gets -1 no matter which policy gets implemented.  If the agent reports $\varnothing$, then no policy is implemented and the principal gets 0.

The type $\theta$ agent gets $0<w_{\theta}<1$ from being right (either proposing the true
cause when there is one, or reporting $\varnothing$ when no true
cause exists in the dataset), and gets $1-w_{\theta}$ when the reported covariate
is implemented. (Note that agents would never abstain from proposing a covariate
even if they had this option in the baseline model or in the previous extensions,
since they always get zero utility from abstaining but expect to get strictly positive utility from proposing a random covariate.)

\begin{prop}
\label{prop:no_true_cause}Suppose $w_{\text{maven}}>3/4$ and $\beta>w_{\text{maven}}$. Then there exists some $\bar{q}>0$
so that the principal strictly prefers any $q$ level of noise with
$0<q\le\bar{q}$ to 0 noise.
\end{prop}
This result says that even when there is some probability that none
of the covariates is the true cause, provided this probability is
not too high and agents put enough weight on being right, a small
enough amount of dissemination noise is still strictly better than
no noise. A stronger assumption on $w_{\text{maven}}$ is needed for this result compared to the previous results. This
 ensures that when the maven is sufficiently confident that neither $a^*$ nor $a^{\redherr}$ is the true cause,
he would rather report $\varnothing$ (and get the utility for being right) than report some wrong covariate
that passes the policymaker's test.

\section{\label{sec:example}A Numerical Example of Dissemination Noise and $p$-Hacking in Linear Regressions}

While our models from Sections \ref{sec:model} and \ref{sec:special} are highly stylized, we show through a numerical example that dissemination noise may play a similar role in screening the researcher's
expertise in more realistic empirical settings that do not satisfy all of our model's simplifying assumptions. 

We consider a linear regression setting with a continuous outcome variable and some continuous covariates, where the outcome  is the sum of three covariates plus noise. The three causal covariates are randomly selected from a set of potential explanatory variables, and the principal
would like to implement a policy that correctly targets the causal covariates. An uninformed agent
may $p$-hack by trying out all regressions involving different triplets of explanatory variables to game the policymaker's test, which simply evaluates the
agent's econometric specification based on its explanatory power on the raw dataset. 

To be concrete, there are 20 covariates $X^{1},\dots,X^{20}$, with each $X^{i}\sim\mathcal{N}(0,1)$,
where $\mathcal{N}(\mu,\sigma^{2})$ is the normal distribution with
mean $\mu$ and variance $\sigma^{2}.$ It is known that the outcome $Y$ is generated
from a linear model $Y=X^{i_{1}^{*}}+X^{i_{2}^{*}}+X^{i_{3}^{*}}+\epsilon$
where $1\le i_{1}^{*}<i_{2}^{*}<i_{3}^{*}\le20$ are three of the
covariates, with all triplets equally likely, and $\epsilon\sim\mathcal{N}(0,4)$ is an error term.
Without loss, we analyze the case when the causal covariates have the realization ($i_{1}^{*},i_{2}^{*},i_{3}^{*})=(1,2,3)$.
The principal's raw dataset
consists of 20 independent observations of the outcome variable and
the covariates from their joint distribution.

The principal disseminates
a noisy version of the data to the agent by adding an independent
noise term with the distribution $\mathcal{N}(0,\sigma_{noise}^{2})$
to every realization of each covariate in the dataset. The noise variance $\sigma_{noise}^{2}$ controls how much dissemination noise the principal injects into the released data.

The agent analyzes the data and proposes a model
$(\hat{\imath}_{1},\hat{\imath}_{2},\hat{\imath}_{3})$ for $Y$.
Then a policymaker tests  the proposed model on the raw data and implements a policy targeting
the covariates $(\hat{\imath}_{1},\hat{\imath}_{2},\hat{\imath}_{3})$
if the model passes the test. Suppose
$(\hat{\imath}_{1},\hat{\imath}_{2},\hat{\imath}_{3})$ passes the
test when the linear regression's $R^{2}$ on the raw data exceeds a critical value,
otherwise the proposal is rejected. The critical
value is the 95-percentile $R^{2}$ when a triplet of covariates is
chosen uniformly at random from all possible triplets. The principal gets utility 1 if the correct specification $(1, 2, 3)$ passes the test, utility -1 if any other specification passes the test, and utility 0 if the proposal is rejected.

With some probability, the agent is a hacker who is uninformed about $(i_{1}^{*},i_{2}^{*},i_{3}^{*})$ and runs all
${20 \choose 3}=1140$ linear regressions of the form $Y=X^{i_{1}}+X^{i_{2}}+X^{i_{3}}+\epsilon$
for different choices of the three covariates $i_{1},i_{2},i_{3}$ in the noisy data. The agent then proposes the model with the highest $R^2$ value. 
With complementary probability, the agent is a maven whose expertise
in the subject lets him narrow down the causal model of $Y$ to either the true $Y=X^{1}+X^{2}+X^{3}+\epsilon$, or
the incorrect model $Y=X^{4}+X^{5}+X^{6}+\epsilon$. The maven runs
two regressions using the noisy data, and proposes either $(1, 2, 3)$ or
$(4, 5, 6)$ to the policymaker, depending on which regression has a
higher $R^{2}.$ (Unlike in the model where we derive optimal behavior for the agents, for this example their behavior are exogenously given.)

We draw a few comparisons between the example and our baseline model from Section \ref{sec:model}. Like our model, the example captures a setting
with a wide dataset, in the sense that there are many more potential specifications (more than 1000) than there are observations in the data (20). The true cause $a^*$ corresponds to the triplet $(1, 2, 3)$, and the red herring $a^{\redherr}$ corresponds to the triplet $(4,5,6)$. Unlike in the baseline model, this example does not feature a continuum of potential specifications or independence between different specifications (since two triplets may share some explanatory variables). Nevertheless, we numerically show that injecting dissemination noise in the form of choosing a strictly positive $\sigma_{noise}^{2}$ has some of the same properties and trade-offs for the principal as in our simple baseline model.

\begin{figure}
\begin{centering}
\includegraphics[scale=0.5]{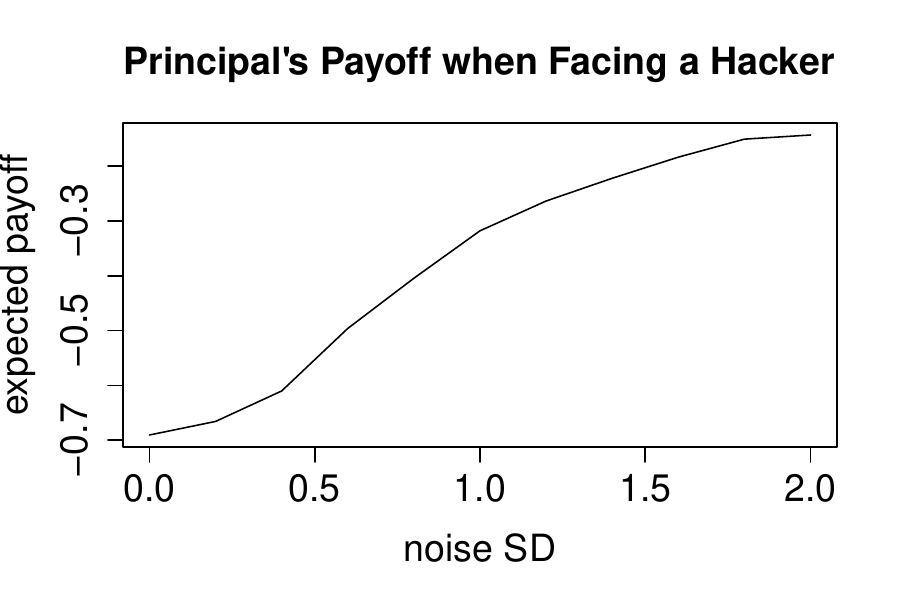} \includegraphics[scale=0.5]{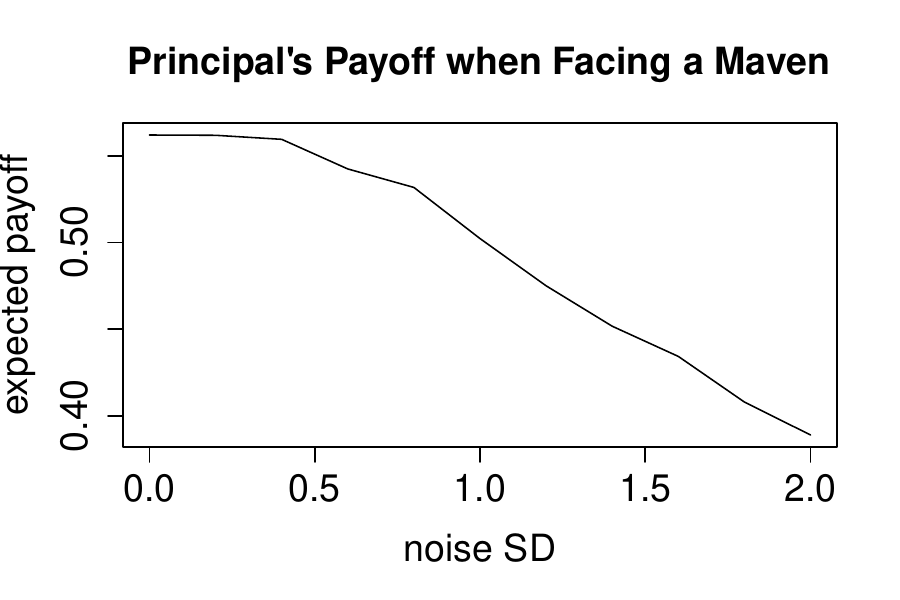} 
\par\end{centering}
\caption{\label{fig:conditional_welfare}The principal's expected utility conditional on the
agent being a hacker or a maven, as a function of the amount of dissemination
noise.}
\end{figure}

Figure \ref{fig:conditional_welfare} depicts the expected utility of the principal
conditional on the agent's type, as a function of the amount of dissemination
noise that the principal adds to the covariates before releasing the
dataset. The expected social harm from a hacker agent is mitigated
when there is more noise. The idea is that when a hacker analyzes
a noisy dataset, the model $(i_{1},i_{2},i_{3})$ with the highest
regression $R^{2}$ in the noisy data is often a \emph{bait} with
poor $R^{2}$ performance in the true dataset. The covariates $i_{1},i_{2},i_{3}$
look correlated with the outcome $Y$ only because they were hit with
just the right noise realizations, but a hacker who falls for these
baits and proposes the model $(i_{1},i_{2},i_{3})$ will get screened
out by the policymaker's test, which is conducted in the raw
data.

Of course, a maven  is also hurt by the noise. The principal's expected payoff when facing a maven falls when  more
dissemination noise is added to data. The maven needs to use the data to compare the
two  candidates (1, 2, 3) and (4, 5, 6). Noisier data makes it harder to identify the true causal model.

\begin{figure}
\begin{centering}
\includegraphics[scale=0.5]{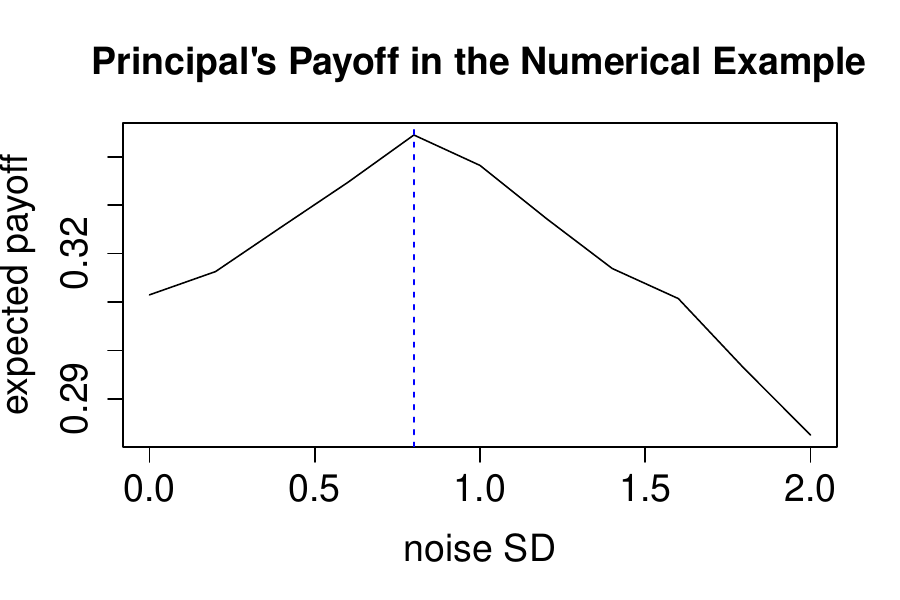} 
\par\end{centering}
\caption{\label{fig:average_payoff}Expected utility of the principal  as a function of
the standard deviation of dissemination noise, when 20\% of the agents are hackers
and 80\% are mavens.}
\end{figure}

Suppose 20\% of the agents are hackers and 80\% are mavens. Figure
\ref{fig:average_payoff} shows the expected social welfare as a function
of the amount of dissemination noise. Consistent with the message of Proposition \ref{prop:mainstatic}, the optimal dissemination noise
trades off screening out hackers, using the baits created by noise,
versus preserving data quality for mavens to identify the correct
model.

The optimal amount of dissemination noise is strictly positive because a small
amount of noise hurts hackers more than mavens. The intuition in this example, as in the model, is that
it is likely that noise creates some baits in the disseminated dataset,
but it is unlikely that the specification  $(4, 5, 6)$ happens to contain one
of the baits. The maven, who only considers the two candidate specification
$(1, 2, 3)$ and $(4, 5, 6)$, is much less likely to fall for a bait than
the hacker, who exhaustively search through all possible specification.

\begin{figure}
\begin{centering}
\includegraphics[scale=0.35]{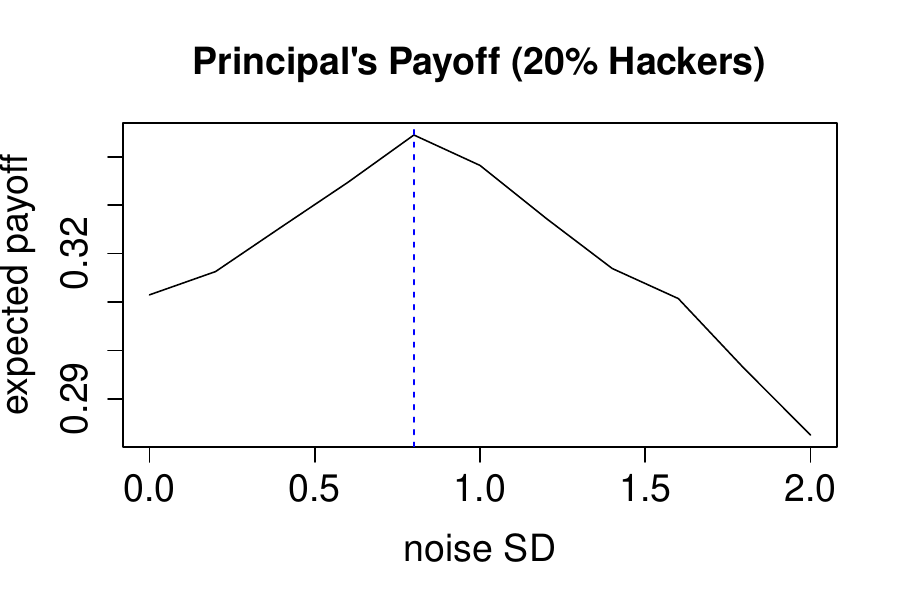} \includegraphics[scale=0.35]{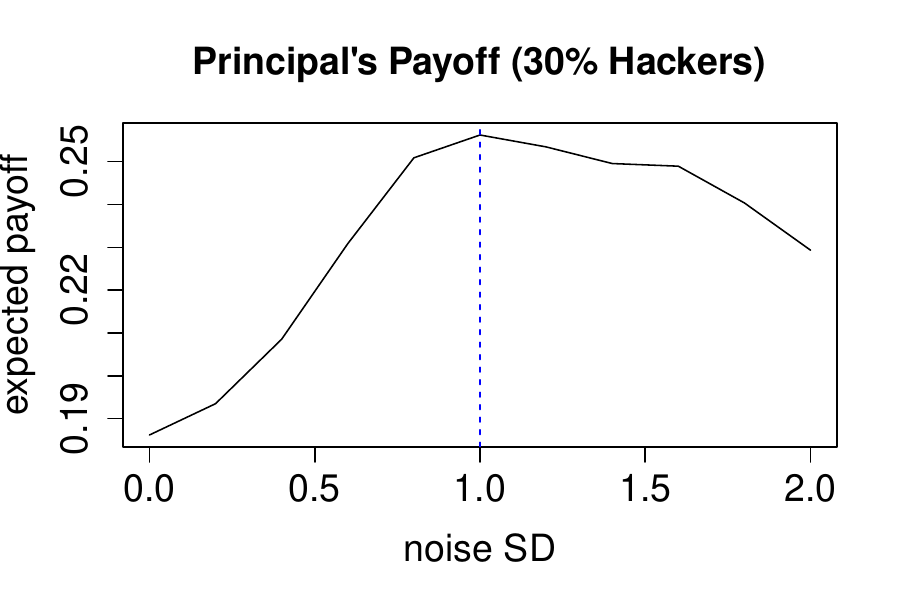}
\includegraphics[scale=0.35]{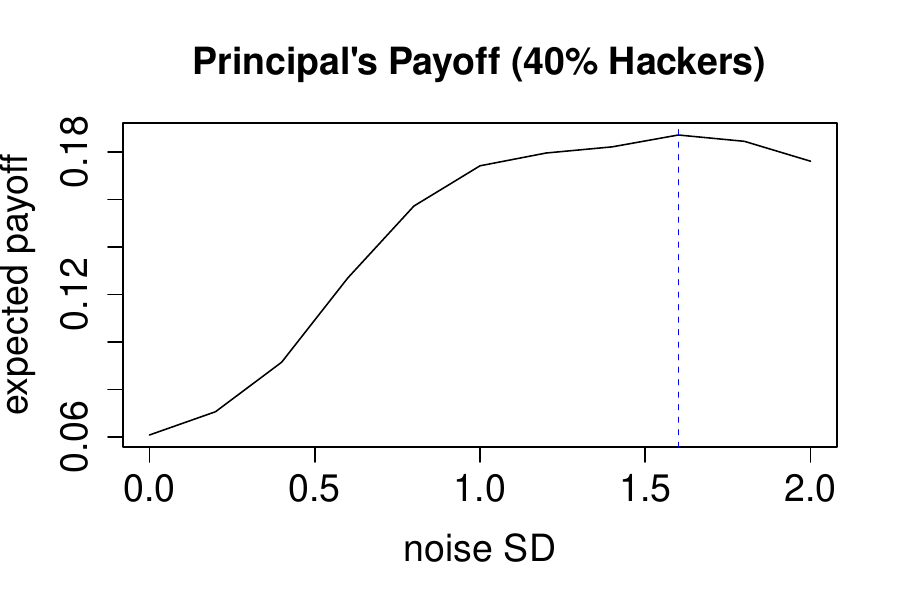}
\end{centering}
\caption{\label{fig:comparative_statics}Comparative statics of the optimal
amount of dissemination noise with respect to the fraction of hackers.}
\end{figure}

\begin{figure}
\begin{centering}
\includegraphics[scale=0.5]{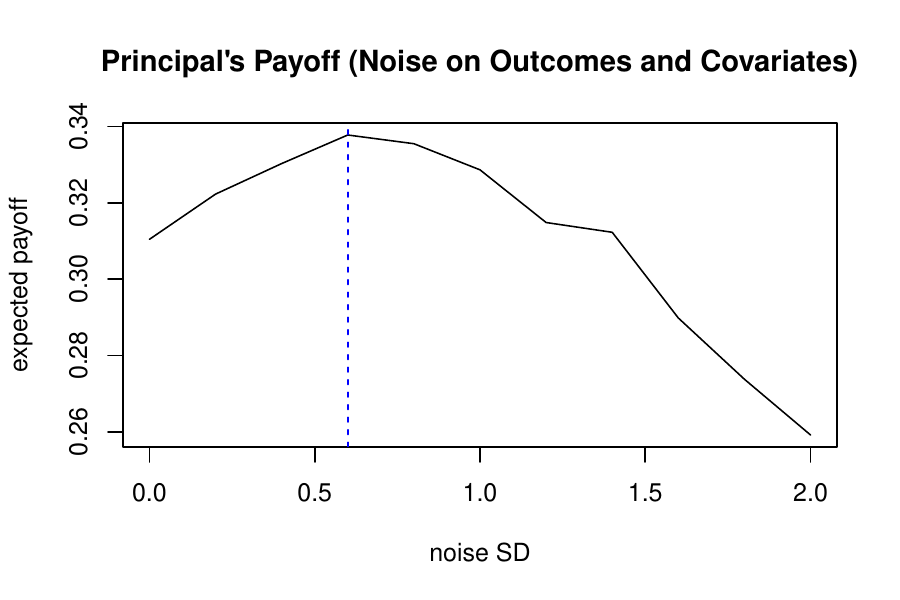} 
\par\end{centering}
\caption{\label{fig:noise_all}Expected utility of the principal  as a function of
the standard deviation of dissemination noise, if noise is added to both the outcome variable and the covariates. Again, we assume that  20\% of the agents are hackers
and 80\% are mavens.}
\end{figure}

Figure \ref{fig:comparative_statics} illustrates the comparative
statics of how the optimal level of dissemination noise varies with
changes in the environment. When the fraction of hackers increases,
more noise is optimal. 

Finally, Figure \ref{fig:noise_all} shows the principal's expected payoff as a function of the  noise standard deviation $\sigma_{noise}$ if dissemination noise is also added to the outcome variable (in addition to the covariates). This numerical result suggests that the optimal amount of noise is lower if, unlike in our theoretical results, noise is added to every variable in the dataset. 

\section{\label{sec:dynamicproblem}Dynamic Environment with Data Reuse}

The models from Sections \ref{sec:model} and \ref{sec:special} presume that a test performed on the raw dataset determines whether the agent's proposal is implemented as policy, but the raw data itself is never publicly revealed. In practice, if these tests determine  high-stakes policy decisions, it may be impossible to credibly conduct them on a secret dataset. Transparency concerns may require that   the data used to test proposals must be made public in a timely manner.\footnote{Indeed secret access is the focus of the work of \citet{dwork2015reusable}, who propose methods for differentiably-private access to the raw data. Their work is motivated by the same concerns over re-use that we turn to in this section. Our results in Section~\ref{sec:dynamicproblem} may be read as validating the use of the methods in \citet{dwork2015reusable}.}

We now turn to a model where the principal owns a dataset that is used by various researchers to study different questions over a long period of time. (For instance, the US Census is only conducted every ten years and the same dataset is used to answer a large number of research questions in the social sciences.) As researchers propose policy interventions to address different research questions, their proposals must all be tested in a transparent way. Suppose the principal is legally bound to periodic releases of noisy data -- multiple noisy ``waves'' of the data are made public over time, which are then used to test the
most recent policy proposals. A policy proposal made in April
2024, for example, would be tested using the May 2024 release of noisy
data. An October proposal would be tested against the November release, and so on. We assume that the principal cannot delay the release of the noisy data used to test past proposals, even though such public data releases become accessible to future $p$-hackers who will try to data mine the same dataset. 

In our dynamic model, time works in the hackers' favor, as $p$-hacking becomes easier when data is reused. Hackers can exploit
all past releases of the noisy data to propose policies that are increasingly
likely to pass the policymaker's test. As we shall see, in the end,
the principal will rationally give up on adding noise to test data, and will release
the original raw dataset. At that point, the hacker can always find
misguided policies that pass the test and get implemented by the policymaker.
The promise of using noisy data to deal with $p$-hacking is real,
but finitely lived.

Time is discrete and infinite:
$t=0,1,2,\ldots$ In period 0, the principal receives a raw dataset
as before, but with the following changes compared to the baseline
static model:
\begin{itemize}
\item The dataset contains a continuum of covariates, $(X^a)_{a \in A}$. But, there is a countably infinite number of outcome variables, $(Y^{t})_{t=0,1,2,...}$.
A true cause $a_{t}^{*}\in A$ is drawn uniformly at random from $A$
for each outcome $Y^{t}$. The principal does not receive more data in later periods: no additional outcomes, covariates, or observations will arrive. The ``dynamic'' aspect of the model concerns how a fixed dataset must be reused over time. 
\item Suppose for simplicity the maven knows
the true cause of every outcome, so red herrings are not generated.
\item For simplicity, suppose there is only a single observation $N=1$
of the  outcomes and the covariates. This is for tractability so that
the state space of the resulting model becomes one-dimensional. It
is, of course, an extreme version of the assumption of a wide dataset.
\item Suppose the unconditional distribution of each outcome variable and
each covariate is Bernoulli($\kappa$) for some $0<\kappa<1$. The
baseline model looked at the case where $\kappa=0.5$.
\end{itemize}

These simplifying changes allow us to focus on the intertemporal trade-offs facing the principal. In each period, she generates a noisy version of the raw dataset to evaluate the agent's proposal. This testing dataset must be publicly released before another agent uses the same dataset to propose the causal covariate behind another outcome variable.  The principal will have a short-term incentive to decrease noise and thus improve the quality of tests for current proposals, but a long-term incentive to increase noise so as to plant baits for future hackers. The intertemporal trade-off will be affected by a ``stock of randomness'' that is decreased as time passes.

In each period, the principal releases a possibly
noisy version of the raw data: in period $t$, she releases a dataset $\mathcal{D}(q_{t})$
after adding a level $q_{t}$ of noise to the raw dataset. The parameter
$q_{t}$ is, as before, the probability that each $X^{a}$ is flipped. (As in the baseline model, the principal only perturbs covariates, not outcome variables.) 
Each release is a \emph{testing dataset}. Note that the principal always adds noise to the raw dataset, not to the previous iteration of the noisy dataset. 

In each period $t=1,2,\ldots$, society is interested in enacting
a policy to target the true cause behind the outcome $Y^{m(t)}$,
where $m(t)$ is the $t$-th outcome with a realization of 1 in the
principal's dataset. So, in the dynamic model we interpret an outcome
realization of 0 as benign, and an outcome realization of 1 as problematic
and requiring intervention. A short-lived agent arrives in each  period $t$; the agent is a hacker with probability $h$ and a maven with complementary probability.  If the agent is a maven, recall that we are assuming the agent
always knows and proposes the true cause of $Y^{m(t)}$. If the agent
is a hacker, he uses all of the testing datasets released by the principal
up to time $t-1$ to make a proposal that maximizes the probability of being implemented. After receiving the agent's proposal
$a$, the principal generates and publishes period $t$'s testing
dataset $\mathcal{D}(q_{t})$. The policymaker implements policy $a$
if $Y^{m(t)}=\hat{X}^{a}$ in this period's (possibly noisy) testing
dataset. In period $t,$ the principal gets a payoff of 1 if the true
cause for $Y^{m(t)}$ passes the test, -1 if any other covariate passes
the test, and 0 if the proposal is rejected. The principal maximizes
expected discounted utility with discount factor $\delta\in (0,1)$

In each period $t\ge2$, a hacker proposes a policy $a$ with $\hat{X}^{a}=1$
in all of the past testing datasets. Such $a$ exists because
there are infinitely many policies. (In the first period, the hacker
has no information and proposes a policy uniformly at random.) Suppose
a covariate $a$ that shows as ``1'' in all the noisy testing datasets
up to period $t-1$ has some $b_{t}$ chance of being a \emph{bait}, that
is $X^{a}\ne1$ in the raw data. Then the principal's expected payoff today
from releasing a testing dataset with noise level $q_t$ is 
\[
u(q_{t};b_{t}):=(1-h)(1-q_{t})+h(-(1-b_{t})(1-q_{t})-b_{t}q_{t}).
\]
In the expression for $u$, $(1-h)(1-q_{t})$ is the probability that
the agent is a maven and the value of the true cause for $Y^{t}$
in the period $t$ testing dataset, $\hat{X}^{a_{t}^{*}}$, has not
been flipped. The term $(1-b_{t})(1-q_{t})$ represents the probability
that the hacker's policy is not a bait and its covariate value has
not been flipped in the testing dataset. Finally, $b_{t}q_{t}$ is
the probability that the hacker's policy is a bait, but its covariate
value has been flipped in the testing dataset.

The principal's problem is similar to an intertemporal consumption
problem. We can think of $b_{t}$ as a stock variable that gets consumed
over time. But rather than a stock of some physical capital, it measures
the \emph{stock of randomness} in the principal's raw dataset. This stock 
depletes as more and more noisy versions of the data are made public.
We view $u(q;b)$ as the principal's flow utility from ``consuming''
$\frac{1}{2}-q$, where the stock of randomness left
is $b$, and the stock evolves according to $b_{t+1}=\frac{b_{t}q_{t}}{(1-b_{t})(1-q_{t})+b_{t}q_{t}}$.

The intertemporal trade-offs faced by the principal are captured by  $\frac{\partial u}{\partial q}<0,$ $\frac{\partial  u}{\partial b}>0,$ 
and $\frac{\partial b_{t+1}}{\partial q_{t}}>0.$ In words, adding less noise
to the testing dataset today gives higher utility today, since the
maven's correct policy is more likely to pass the test, and the hacker's
misguided policy is more likely to get screened out. But this depletes
the stock of randomness faster, and makes it harder to defend against
future hackers.

Our next result shows that, in every optimal solution to the principal's
problem, the stock of randomness is always depleted in finite time.
The basic idea is that noise has decreasing returns: the marginal
effect of noise on slowing the decline of $b_{t+1}$ is reduced as
$b_{t}$ decreases. There is a time $t^{*}$ at which the principal
abandons the use of noise. 
\begin{prop}
\label{prop:dynamic} Suppose that $h<1/2$ and $\kappa \in(0,1)$. Let
$\{(b_{t},q_{t})\}$ be a solution to the principal's problem. Then,
for all $t$, $q_{t}<1/2$ and $b_{t}$ is (weakly) monotonically
decreasing. There $t^{*}$ such that
\begin{itemize}
\item If $t<t^{*}$ then $b_{t+1}<b_{t}$;
\item If $t\geq t^{*}$ then $q_{t}=0$ and $b_{t+1}=0$.
\end{itemize}
\begin{figure}
\begin{centering}
\includegraphics[scale=0.7]{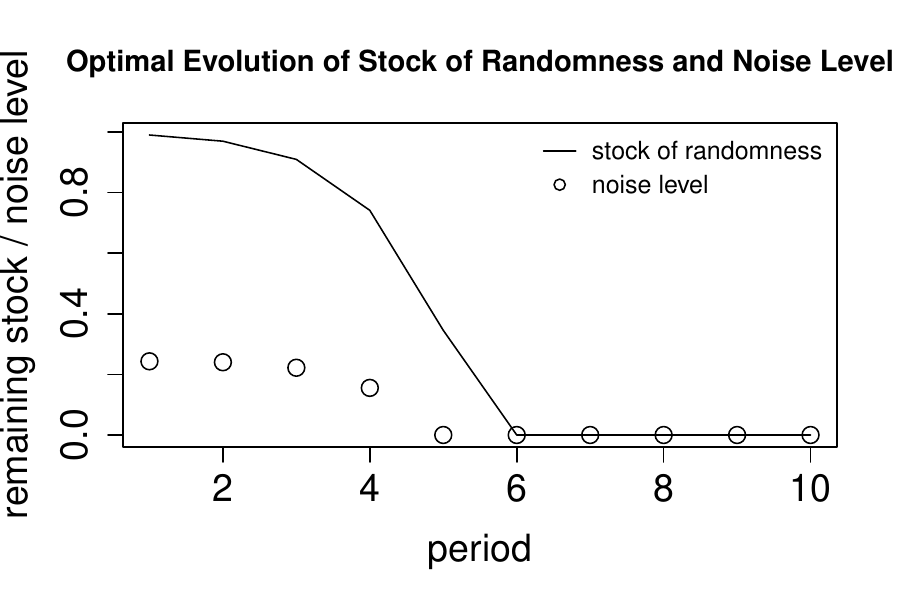}
\par\end{centering}
\caption{\label{fig:dynamic_solution}The evolution of the stock of randomness
(i.e., the probability $b_{t}$ that the hacker's best guess $a$
is a bait with $X^{a}\protect\ne1$ in the raw dataset) and the noise
level in an environment with 45\% hackers, discount factor $\delta=0.99,$
and an unconditional probability $\kappa=0.01$ that each covariate
is equal to 1.}
\end{figure}
\end{prop}
Figure \ref{fig:dynamic_solution} shows an example with $\kappa=0.01$,
$\delta=0.99,$ and $h=0.45.$ In period 1, a hacker has a 1\% chance
of guessing a covariate that would validate in the raw dataset. The
principal releases noisy testing datasets at the end of periods 1,
2, 3, and 4. In period 5, a hacker can look for a covariate that has
a value of ``1'' in each of the four testing datasets from the previous
periods, and propose it as the model for today's outcome variable
$Y^{m(5)}.$ This proposal will validate in the raw dataset with more
than 65\% probability, reflecting a weakening defense against $p$-hackers
as data is reused and the stock of randomness depletes. At this point,
the principal finds it optimal to give up on dissemination noise and
releases the raw dataset as the testing dataset at the end of period
5. In every subsequent period, both agent types will propose passing
policies, so the policymaker implements correct policies 55\% of the
time and misguided policies 45\% of the time.

\section{Concluding Discussion}
\label{sec:conclusion}
We argue that infusing data with noise before making data public  has benefits beyond the privacy protection guarantees for which the practice is currently being used. Noise baits uninformed $p$-hackers into finding correlations that can be shown to be spurious. The paper investigates these ideas in a simple model that captures the trade-off between preventing hackers from passing off false findings as true, and harming legitimate research that seeks to test an ex-ante  hypothesis.

\bibliographystyle{ecta}
\bibliography{p_hacking}
\appendix
\section{Proofs}

\subsection{Proof of Lemma \ref{lem:agent_behavior_new}}
\begin{proof}
For any covariate $a \ne a^*$, we have $\mathbb{P}[X_{n}^{a}=Y_{n}\mid\hat{X}_{n}^{a}=Y_{n}]=1-q$
and $\mathbb{P}[X_{n}^{a}=Y_{n}\mid\hat{X}_{n}^{a}\ne Y_{n}]=q$ for each observation $n$. Since $1-q>q,$
the distribution of the number of observations that match the outcome
in the raw data for a covariate $\hat{a}$ with $\hat{X}_{n}^{\hat{a}}=Y_{n}$ for every $n$
first-order stochastically dominates that of any other covariate. Reporting such a covariate thus maximizes the chance of passing the policymaker's test, for any $\gamma$. Since the hacker will never find $a^*$, he will only maximize this passing chance.  Similarly, among those covariates $\hat{a}\notin\{a^{*},a^{\redherr}\}$, the
maven's expected payoff is maximized by covariates $\hat{a}$ with
$\hat{X}_{n}^{\hat{a}}=Y_{n}$ for every $n$.

It remains to show how the maven optimally chooses between $a^{*}$
and $a^{\redherr}$ based on his information.  Suppose the maven learns the true
cause and the red herring are in the set $\{a',a''\},$ where $\hat{X}^{a'}$
matches $Y$ in $k_{1}$ observations and $\hat{X}^{a''}$ matches
$Y$ in $k_{2}$ observations with $k_{1}<k_{2}$.


If $a^{*}=a'$, the maven's data has likelihood $\tilde{p}^{k_{1}}(1-\tilde{p})^{N-k_{1}}\cdot(1/2)^{N}$
where $\tilde{p}=\psi(1-q)+(1-\psi)q>1/2.$ This is because  $\mathbb{P}[\hat{X}_{n}^{a^{*}}=Y_{n}]=\tilde{p}$
for every observation $n$, while $\mathbb{P}[\hat{X}_{n}^{a^{\redherr}}=Y_{n}]=1/2$
for every observation $n$. If $a^{*}=a''$, then the data likelihood
is $\tilde{p}^{k_{2}}(1-\tilde{p})^{N-k_{2}}\cdot(1/2)^{N}$ by the
same reasoning. This second likelihood is larger, because $k_{2}>k_{1}$
and $\tilde{p}>1/2.$ So the covariate that matches the outcome in
more observations has a higher posterior probability of being the
true cause.

Let $a\in\{a',a''\}$ and consider the probability that $a$ passes
the policymaker's test. Conditional on $a=a^{*}$, if $\hat{X}_{n}^{a}=Y_{n}$, then there
is $\frac{\psi(1-q)}{\psi(1-q)+(1-\psi)q}$ chance that $X_{n}^{a}=Y_{n}$.
If $\hat{X}_{n}^{a}\ne Y_{n}$, then there is $\frac{\psi q}{\psi q+(1-\psi)(1-q)}$
chance that $X_{n}^{a}=Y_{n}$. Conditional on $a=a^{\redherr}$,
if $\hat{X}_{n}^{a}=Y_{n}$, then there is $1-q$ chance that $X_{n}^{a}=Y_{n}$.
If  $\hat{X}_{n}^{a}\ne Y_{n}$, then there is $q$ chance that
$X_{n}^{a}=Y_{n}$. By simple algebra, $\psi>1/2$ implies that $\frac{\psi(1-q)}{\psi(1-q)+(1-\psi)q}>1-q$
and $\frac{\psi q}{\psi q+(1-\psi)(1-q)}>q.$ So, the covariate $a''$
with a higher posterior probability of being the true cause also has
a higher chance of passing the policymaker's test. The maven gets
strictly higher expected utility from proposing $a''$ than $a'$ because the probability of being right, and of passing the test, are higher under $a''$ than under $a'$.
\end{proof}

\subsection{Proof of Proposition \ref{prop:asymptotic_general}}
\begin{proof}
Let such a $q$ be fixed. 

When the hacker proposes $\hat{a}$ with $\hat{X}_{n}^{\hat{a}}=Y_{n}$
for every $1\le n\le N$, the number of observations $n$ where $X_{n}^{\hat{a}}=Y_{n}$
has the distribution $\text{Binom}(N,1-q).$ Let $\ell_N$ be the probability that the hacker's proposal is accepted. Since $1-q<\gamma,$ we have  $\ell_N \to 0$ as
$N\to\infty$. So the principal's utility conditional on facing a
maven converges to 0 when $N\to\infty.$

For the maven, consider the following three events. 

$E_{1}$: $|\{n:\hat{X}_{n}^{a^{*}}=Y_{n}\}|\le|\{n:\hat{X}_n^{a^{\redherr}}=Y_{n}\}|$.
The distribution of the LHS is $\text{Binom}(N,\psi(1-q)+(1-\psi)q)$,
where $\psi(1-q)+(1-\psi)q>1/2$ since $\psi>1/2$ and $q<1/2.$ Also,
the distribution of the RHS is $\text{Binom}(N,1/2).$ By the law of large
numbers, the probability of $E_{1}$ goes to 0 when $N\to\infty$.

$E_{2}:|\{n: X_{n}^{a^{*}}=Y_{n}\}|<\left\lfloor \gamma\cdot N\right\rfloor $.
The LHS has the distribution $\text{Binom}(N,\psi)$ where $\psi>\gamma$.
So by the law of large numbers, the probability of $E_{2}$ also approaches
0 when $N\to\infty.$

Let $E_{3}$ be the event that the maven optimally proposes a covariate
other than $a^{*}$ or $a^{\redherr}$. Consider the feasible strategy
for the maven of picking between the two covariates identified by his private information
based on which one matches the outcome in more observations. Except on the event
$E_{1}\cup E_{2}$, this feasible strategy correctly identifies $a^{*}$
and the proposal gets accepted by the policymaker. Therefore, the
expected payoff of this feasible strategy for the maven is at least
$1-\mathbb{P}[E_{1}\cup E_{2}]$ for any value of $0\le w_{\text{maven}}\le1.$
When the maven proposes some $\hat{a}\notin\{a^{*},a^{\redherr}\}$
with $\hat{X}_{n}^{\hat{a}}=Y_{n}$ for every $n$, he gets $1-w_{\text{maven}}$
if the proposal is accepted and 0 otherwise, and acceptance happens
with probability $\ell_{N}$. In order for the maven's optimal strategy
to achieve at least an expected payoff of $1-\mathbb{P}[E_{1}\cup E_{2}]$,
we therefore need 
\[
\mathbb{P}[E_{3}]\cdot(1-w_{\text{maven}})\cdot\ell_{N}+(1-\mathbb{P}[E_{3}])\cdot1\ge1-\mathbb{P}[E_{1}\cup E_{2}]\iff\mathbb{P}[E_{3}]\le\frac{\mathbb{P}[E_{1}\cup E_{2}]}{1-(1-w_{\text{maven}})\ell_{N}}.
\]
Note that we have $\mathbb{P}[E_{3}]\to0$ as $N\to\infty,$ since
we have both $\mathbb{P}[E_{1}\cup E_{2}]\to0$ and $\ell_{N}\to0$
as $N\to\infty$.

Outside of the event $E_{1}\cup E_{2}\cup E_{3}$, the maven's optimal
strategy proposes $a^{*}$ and this proposal passes the policymaker's
test. So the principal's utility conditional on facing a maven converges
to $1$ when $N\to\infty.$

\end{proof}

\subsection{Proof of Proposition \ref{prop:N_bound}}

\begin{proof}

Consider the events $E_{1},E_{2},$ and $E_{3}$ in the proof of Proposition
\ref{prop:asymptotic_general}.

First, we bound the probability of $E_{1}$. Using Hoeffding's inequality,
since $|\{n:\hat{X}_{n}^{a^{*}}=Y_{n}\}|$ has a binomial distribution
with success rate $\psi(1-q)+(1-\psi)q$, 
\[
\mathbb{P}[|\{n:\hat{X}_{n}^{a^{*}}=Y_{n}\}|\le\frac{\psi(1-q)+(1-\psi)q+0.5}{2}N]\le\exp(-2N[\frac{\psi(1-q)+(1-\psi)q-0.5}{2}]^{2}).
\]
Similarly, since $|\{n:\hat{X}_{n}^{a^{\redherr}}=Y_{n}\}|$ has a
binomial distribution with success rate 1/2, 
\[
\mathbb{P}[|\{n:\hat{X}_{n}^{a^{\redherr}}=Y_{n}\}|>\frac{\psi(1-q)+(1-\psi)q+0.5}{2}N]\le\exp(-2N[\frac{\psi(1-q)+(1-\psi)q-0.5}{2}]^{2}).
\]
Bounding the probability of each of these two events by $h/32$ requires
$N\ge\frac{-2\ln(h/32)}{(\psi(1-q)+(1-\psi)q-0.5)^{2}}$. This ensures
$\mathbb{P}[E_{1}]\le h/16.$

Also by Hoeffding's inequality, $\mathbb{P}[E_{2}]\le\exp(-2N[\psi-\gamma]^{2})$.
So whenever $N\ge\frac{-\ln(h/16)}{2(\psi-\gamma)^{2}},$ we have
$\mathbb{P}[E_{2}]\le h/16.$

The probability that the hacker's proposal gets accepted is bounded
by $\exp(-2N\cdot[q-(1-\gamma)]^{2})$ by Hoeffding's inequality.
This quantity is less than $h/8$ whenever $N\ge\frac{-\ln(h/8)}{2\cdot[q-(1-\gamma)]^{2}}.$

From the proof of Proposition \ref{prop:asymptotic_general}, we have
$\mathbb{P}[E_{3}]\le\frac{\mathbb{P}[E_{1}\cup E_{2}]}{1-(1-w_{\text{maven}})\ell_{N}}\le\frac{\mathbb{P}[E_{1}\cup E_{2}]}{1-\ell_{N}}$.
When $\mathbb{P}[E_{1}\cup E_{2}]\le h/8$ and $\ell_{N}\le h/8$,
we have $\mathbb{P}[E_{3}]\le h/4.$

So whenever $N$ satisfies the bound in the statement of the proposition,
we have $\mathbb{P}[E_{1}\cup E_{2}\cup E_{3}]\le(3h)/8$. The principal's
expected payoff is at least $-1\cdot(\frac{3h}{8})+(1-\frac{3h}{8})$
when the agent is a maven, and at least $-1\cdot(h/8)$ when the agent
is a hacker. So, the principal's total expected payoff is larger than
$(1-h)\cdot[-1\cdot(\frac{3h}{8})+(1-\frac{3h}{8})]+h\cdot(-h/8)\ge1-\frac{7h}{4}-\frac{h^{2}}{8}\ge1-2h.$
By comparison, the principal's payoff is smaller than $1-2h$ when
$q=0.$

\end{proof}

\subsection{Proof of Lemma \ref{lem:agent_behavior}}
\begin{proof}
Any strategy of the hacker leads to zero probability of proposing
the true cause, so the hacker finds it optimal to just maximize the
probability of the proposal passing the test. If the hacker proposes
a covariate that matches $Y$ in $n_{1}$ observations and mismatches
in $n_{0}$ observations, then the distribution of the number of matches
in the raw dataset is $\text{Binom}(n_{0},q)+\text{Binom}(n_{1},1-q).$
A covariate that matches the outcome variable in every observation
in noisy dataset will have a distribution of $\text{Binom}(n_{0}+n_{1},1-q)=\text{Binom}(n_{0},1-q)+\text{Binom}(n_{1},1-q)$
as its number of matches in the raw dataset, and $\text{Binom}(n_{0},1-q)$
strictly first-order stochastic dominates $\text{Binom}(n_{0},q)$
if $n_{0}\ge1$ and $q<1/2.$ Therefore the hacker finds it optimal
to propose any $a\in A$ that satisfies $\hat{X}_{n}^{a}=Y_{n}$ for
every $1\le n\le N$.

For the maven, since $w_{\text{maven}} >  1/2$,
it is never optimal to propose covariates other than $a^{*}$ or $a^{\redherr}$
since these have zero chance of being the true cause. Out of the two
candidate covariates that the maven narrows down to, the one that
matches $Y$ in more observations in the noisy dataset has a higher
posterior probability of being the true cause. Note that if the maven
proposes $a^{\redherr}$, the policymaker always rejects the proposal
since $X^{a^{\redherr}}=1-Y$ in the raw dataset. Also, if the maven
proposes $a^{*}$, it always passes the test since $X^{a^{*}}=Y$
in the raw dataset.
\end{proof}

\subsection{Proof of Lemma \ref{lem:lemmarginalutility}}
\begin{proof}
The hacker picks a covariate $\hat{a}$ where $\hat{X}_{n}^{\hat{a}}=Y_{n}$
for every $n$. Given that the policymaker is using the most stringent
test, we get $V_{\text{hacker}}(q)=(1-q)^{N}$. For the maven, there are $2N$ bits of observations on the variables
$X^{a^{*}}$ and $X^{a^{\redherr}}$. If strictly fewer than $N$
bits are flipped, then the maven proposes the correct policy (and therefore it passes the test). If
exactly $N$ bits are flipped, then the maven recommends the correct
policy 1/2 of the time. So,

\[
V_{\text{maven}}(q)=(\mathbb{P}[\text{Binom}(2N,q)<N]+\frac{1}{2}\mathbb{P}[\text{Binom}(2N,q)=N])
\]

We have 
\begin{align*}
V'_{\text{maven}}(q)= & \frac{d}{dq}(\mathbb{P}[\text{Binom}(2N,q)<N]+\frac{1}{2}\mathbb{P}[\text{Binom}(2N,q)=N])\\
= & \frac{d}{dq}(\mathbb{P}[\text{Binom}(2N,q)\le N]-\frac{1}{2}\mathbb{P}[\text{Binom}(2N,q)=N])\\
= & -2N\cdot\mathbb{P}[\text{Binom}(2N-1,q)=N]-\frac{1}{2}\frac{d}{dq}(q^{N}(1-q)^{N}\binom{2N}{N}),
\end{align*}
where the last step used the identity that $\frac{d}{dq}\mathbb{P}[\text{Binom}(M,q)\le N]=-M\cdot\mathbb{P}[\text{Binom}(M-1,q)=N]$.
Continuing, 
\begin{align*}
 & -2N\cdot q^{N}(1-q)^{N-1}\binom{2N-1}{N}-\frac{1}{2}\binom{2N}{N}N(q^{N-1}(1-q)^{N}-q^{N}(1-q)^{N-1})\\
= & \binom{2N-1}{N}Nq^{N-1}(1-q)^{N-1}\left(-2q-\frac{1}{2}\cdot2\cdot((1-q)-q)\right),
\end{align*}
using the identity $\binom{2N}{N}=2\cdot\binom{2N-1}{N}$. Rearranging
shows the lemma. 
\end{proof}

\subsection{Proof of Proposition \ref{prop:mainstatic}}
\begin{proof}
Using the Lemma \ref{lem:lemmarginalutility}, 
\[
\frac{d}{dq}[-hV_{\text{hacker}}(q)+(1-h)V_{\text{maven}}(q)]=hN(1-q)^{N-1}-(1-h)\binom{2N-1}{N}Nq^{N-1}(1-q)^{N-1}.
\]
The FOC sets this to 0, so $h-(1-h)\binom{2N-1}{N}q^{N-1}=0$. Rearranging
gives $q^{*}=\left(\frac{h}{1-h}\frac{1}{\binom{2N-1}{N}}\right)^{1/(N-1)}.$
We know $h\mapsto\frac{h}{1-h}$ is increasing, so $\frac{\partial q^{*}}{\partial h}>0$.
We know $N\mapsto\binom{2N-1}{N}$ is increasing in $N$, therefore
both the base and the exponent in $q^{*}$ decrease in $N$, so $\frac{\partial q^{*}}{\partial N}<0.$
\end{proof}

\subsection{Proof of Proposition \ref{prop:optimal_threshold}}
\begin{proof}
The principal's expected utility conditional on the agent being a
maven is the same for every $\underline{N}\in\{1,...,N\}$, since
the maven always proposes either $a^{*}$ or $a^{\redherr}$ depending
on which covariate matches $Y$ in more observations, and the proposal
passes the $\underline{N}$ threshold if and only if it is $a^{*}$,
since $X^{a^{\redherr}}=1-Y$ does not match the outcome in any observation
in the raw dataset.

As shown in the proof of Lemma \ref{lem:agent_behavior}, the distribution
of the number of matches between $X^{a}$ and $Y$ in the raw dataset
increases in the first-order stochastic sense with the number of matches
between $\hat{X}^{a}$ and $Y$ in the noisy dataset. So, for any
test threshold $\underline{N},$ the hacker finds it optimal to propose
a covariate $a$ with $\hat{X}_{n}^{a}=Y_{n}$ for every $n.$

Therefore, the only effect of lowering $\underline{N}$ from $N$
is to increase the probability of the hacker's misguided policies
passing the test.
\end{proof}

\subsection{Proof of Lemma \ref{lem:hacker_non_iid}}
\begin{proof}
Let $y\in\{0,1\}^{N}$ and $q\in[0,1/2]$ be given. Let $\mu_{q}\in\Delta(\{0,1\}^{N}\}$
be the distribution of covariate realizations in the noisy dataset
with $q$ level of noise. We have $\mathbb{P}[X^{a}=y\mid\hat{X}^{a}=y]=\frac{(1-q)^{N}\cdot\mu(y)}{\mu_{q}(y)}$.
Also, for any $x\in\{0,1\}^{N}$ so that $y$ and $x$ differ in $k$
of the $N$ coordinates, we have $\mathbb{P}[X^{a}=y\mid\hat{X}^{a}=x]=\frac{(1-q)^{N-k}q^{k}\cdot\mu(y)}{\mu_{q}(x)}$.
Note that 
\[
\mu_{q}(x)=\sum_{z\in\{0,1\}^{N}}\mu(z)\cdot q^{D(z,x)}(1-q)^{N-D(z,x)}
\]
where $D(z,x)$ is the number of coordinates where $z$ differs from
$x.$ By the triangle inequality, $D(z,y)\le D(z,x)+D(x,y)=D(z,x)+k.$
This shows for every $z\in\{0,1\}^{N},$ 
\[
q^{D(z,y)}(1-q)^{N-D(z,y)}\ge q^{D(z,x)}(1-q)^{N-D(z,x)}\cdot(\frac{q}{1-q})^{k}.
\]
 So, 
\[
\mu_{q}(x)\ge(\frac{q}{1-q})^{k}\cdot\sum_{z\in\{0,1\}^{N}}\mu(z)\cdot q^{D(z,y)}(1-q)^{N-D(z,y)}=(\frac{q}{1-q})^{k}\mu_{q}(y).
\]
 This shows 
\[
\frac{(1-q)^{N}\cdot\mu(y)}{\mu_{q}(y)}\ge\frac{(1-q)^{N}\cdot\mu(y)}{\mu_{q}(x)\cdot(\frac{1-q}{q})^{k}}=\frac{(1-q)^{N-k}q^{k}\cdot\mu(y)}{\mu_{q}(x)}.
\]
\end{proof}

\subsection{Proof of Proposition \ref{prop:noniid}}
\begin{proof}
First, observe the maven will choose the covariate $a\in\{a^{*},a^{\redherr}\}$
whose noisy realization $\hat{X}^{a}$ matches the outcome $Y$ in
more observations, regardless of $\mu$. This is because the maven
learns two candidates $a_{1},a_{2}\in A$ and knows either $(X^{a_{1}}=Y,X^{a_{2}}=1-Y)$
or $(X^{a_{1}}=1-Y,X^{a_{2}}=Y)$, equally likely. The likelihood
of the former is $\frac{1}{2}\cdot q^{(N-m_{1})+m_{2}}(1-q)^{m_{1}+(N-m_{2})}$
and the likelihood of the latter is $\frac{1}{2}\cdot q^{m_{1}+(N-m_{2})}(1-q)^{(N-m_{1})+m_{2}}$,
where $m_{1},m_{2}$ count the numbers of observations $n$ where
$\hat{X}_{n}^{a_{1}}=Y_{n}$ an $\hat{X}_{n}^{a_{2}}=Y_{n}$, respectively.
Since $q\in[0,1/2],$ the first likelihood is larger if $m_{1}>m_{2},$
and vice versa. Also, maven's proposal is $a^{*}$ if and only if
it passes the policymaker's test. Thus we see that for any $\mu,$
$V_{\text{maven}}(q)$ is the same as when the observations are i.i.d.

Given the hacker's behavior in Lemma \ref{lem:hacker_non_iid}, to
prove $V_{\text{hacker}}^{'}(0)<0$ it suffices to show that for every
$y\in\{0,1\}^{N}$ and $\mu,$ we have $\frac{\partial}{\partial q}\left[\mathbb{P}[X^{a}=y\mid\hat{X}^{a}=y]\right]_{q=0}<0$.
For $z,x\in\{0,1\}^{N},$ let $D(z,x)$ count the number of coordinates
where $z$ differs from $x.$ Let $\mu_{q}\in\Delta(\{0,1\}^{N}\}$
be the distribution of covariate realizations in the noisy dataset
with $q$ level of noise. We may write (using the fact $N\ge2$) that
$\mu_{q}(y)=\mu(y)\cdot(1-q)^{N}+\mu(z:D(z,y)=1)\cdot(1-q)^{N-1}q+f(q^{2})$
where $f(q^{2})$ is a polynomial expression where every term contains
at least the second power of $q.$ Therefore, $\frac{\partial}{\partial q}\left[\frac{(1-q)^{N}\mu(y)}{\mu_{q}(y)}\right]_{q=0}$
is:{\small{}
\[
\mu(y)\cdot\left[\frac{-N(1-q)^{N-1}\mu_{q}(y)-(1-q)^{N}\cdot[-N\mu(y)(1-q)^{N-1}+\mu(z:D(z,y)=1)\cdot((1-q)^{N-1}+g(q))]}{(\mu_{q}(y))^{2}}\right]_{q=0}
\]
}{\small\par}

where $f(0)=0.$ Evaluating, we get $\mu(y)\cdot\frac{-N\mu(y)-[-N\mu(y)+\mu(z:D(z,y)=1)]}{(\mu(y))^{2}}=-\frac{\mu(z:D(z,y)=1)}{(\mu(y))}$.
Since $\mu$ has full support, both the numerator and the denominator
are strictly positive, so $\frac{\partial}{\partial q}\left[\mathbb{P}[X^{a}=y\mid\hat{X}^{a}=y]\right]_{q=0}<0.$
\end{proof}

\subsection{Proof of Proposition \ref{prop:generalredherring}}

In this proof we adopt the following notation: we write $d_{Y}$ for
the realized vector $Y_{n}$, $d^{a}$ for the realized vector $X_{n}^{a}$,
for the $a$th covariate. In the noisy data, we use $\tilde{d}^{a}$
for the realization of the noisy version of the $a$ covariate. As
in other results, it is without loss to analyze the case where $d_{Y}=\boldsymbol{1}$,
so the policy maker will only accept a proposal $a$ if it satisfies
that $d^{a}=\boldsymbol{1}$ in the raw data.

First, we derive the posterior probability of $d^{a}=\boldsymbol{1}$
given a realization of $\tilde{d}^{a}$ in the noisy dataset, and
the resulting behavior of the hacker and the maven. The $n$ component
of $\tilde{d}^{a}$ is denoted $\tilde{d}_{n}^{a}.$
\begin{lem}
\label{lem:prob_pass}Suppose that $\tilde{d}^{a}$ satisfies $\sum_{n}\tilde{d}_{n}^{a}=k$.
We have $\mathbb{P}[d^{a}=\boldsymbol{1}\mid\tilde{d}^{a}]=(1-q)^{k}(q)^{N-k}$.
In particular, the hacker chooses some action $a$ with $\tilde{d}^{a}=\boldsymbol{1}$,
and the maven chooses the policy with the higher number of 1's among
$\tilde{d}^{a^{*}}$ and $\tilde{d}^{a^{\redherr}}$. 
\end{lem}
\begin{proof}
Consider any $\tilde{d}^{a}$ with $\sum_{n}\tilde{d}_{n}^{a}=k$.
In the noisy dataset, for any $q,$ every vector in $\{0,1\}^{N}$
is equally likely. So the probability of the data for policy $a$
having realization $\tilde{d}^{a}$ is $2^{-N}.$ The probability
of this realization in the noisy data and the realization being $d^{a}=\boldsymbol{1}$
in $\mathcal{D}$ is $2^{-N}\cdot(1-q)^{k}(q)^{N-k}.$ So the posterior
probability is $(1-q)^{k}(q)^{N-k}$.

The hacker chooses an action $a$ as to maximize $\mathbb{P}[d^{a}=\boldsymbol{1}\mid\tilde{d}^{a}].$
The term $(1-q)^{k}(q)^{N-k}$ is maximized when $k=N,$ since $0\le q\le1/2.$

The maven sees vectors with $k_{1},k_{2}$ numbers of 1's. The likelihood
of the data given the first action is the correct one is $(1-q)^{k_{1}}(q)^{N-k_{1}}\cdot2^{-N}$
(since all vectors are equally likely in the noisy dataset conditional
on $Y=\boldsymbol{1}$, for $a\ne a^{*}$). This is larger than $(1-q)^{k_{2}}(q)^{N-k_{2}}\cdot2^{-N}$
when $k_{1}\ge k_{2}.$ 
\end{proof}
Here is the expression for the principal's expected payoff as a function
of $q.$ 
\begin{lem}
\label{lem:generalredherringformulapayoffs} Let $A,C\sim\text{Binom}(1-q,N)$
and $B\sim\text{Binom}(1/2,N)$, mutually independent. The principal's
expected payoff after releasing a noisy dataset $\mathcal{D}(q)$
is 
\[
-h(1-q)^{N}+(1-h)\cdot\left[\sum_{k=0}^{N}\mathbb{P}(A=k)\cdot\left(\mathbb{P}(B<k)+\frac{1}{2}\mathbb{P}(B=k)-2^{-N}(\mathbb{P}(C>k)+\frac{1}{2}\mathbb{P}(C=k))\right)\right].
\]
\end{lem}
\begin{proof}
With probability $h,$ the agent is a hacker. By Lemma \ref{lem:prob_pass},
the hacker recommends a policy $\hat{a}$ with $\tilde{d}^{\hat{a}}=\boldsymbol{1}$,
which has $(1-q)^{N}$ chance of being accepted by the principal due
to $d^{\hat{a}}=\boldsymbol{1}$.

With probability $1-h,$ the agent is a maven. For the maven, $\sum_{n}\tilde{d}_{n}^{a^{*}}\sim\text{Binom}(1-q,N)$
and $\sum_{n}\tilde{d}_{n}^{a^{\redherr}}\sim\text{Binom}(1/2,N)$ are independent.
Whenever $\sum_{n}\tilde{d}_{n}^{a^{*}}>\sum_{n}\tilde{d}_{n}^{a^{\redherr}}$,
and with 50\% probability when $\sum_{n}\tilde{d}_{n}^{a^{*}}>\sum_{n}\tilde{d}_{n}^{a^{\redherr}}$,
the maven recommend $a^{*}$ by Lemma \ref{lem:prob_pass}, which
will be implemented by the principal.

When maven recommends $a^{\redherr}$, the principal only implements it (and
gets utility -1) if $d^{a^{\redherr}}=\boldsymbol{1}.$ The probability of
$d^{a^{\redherr}}=\boldsymbol{1}$ is $2^{-N},$ and the probability of $a^{\redherr}$
being recommended given $d^{a^{\redherr}}=\boldsymbol{1}$ and $\sum_{n}\tilde{d}_{n}^{a^{*}}=k$
is $\mathbb{P}(C>k)+\frac{1}{2}\mathbb{P}(C=k)$, interpreting $C$
as the number of coordinates that did not switch from $d^{a^{\redherr}}$
to $\tilde{d}^{a^{\redherr}}$.
\end{proof}
Now, with the formula for the principal's expected payoff in place,
we can evaluate the derivative at $q=0$ to prove Proposition \ref{prop:generalredherring}.
\begin{proof}
We apply the product rule. First consider $\frac{d}{dq}\mathbb{P}(A=k)|_{q=0}$.

We have $\mathbb{P}(A=k)=(1-q)^{k}q^{N-k}\binom{N}{k}.$ If $k<N-1,$
then every term contains at least $q^{2}$ and its derivative evaluated
at 0 is 0. For $k=N,$ we get $(1-q)^{N}$ whose derivative in $q$
is $-N(1-q)^{N-1}$, which is $-N$ evaluated at 0. For $k=N-1,$
we get $(1-q)^{N-1}q\cdot N$, whose derivative evaluated at 0 is
$N.$

We now evaluate, for $k=N-1,N$: 
\[
\left(\mathbb{P}(B<k)+\frac{1}{2}\mathbb{P}(B=k)-2^{-N}(\mathbb{P}(C>k)+\frac{1}{2}\mathbb{P}(C=k))\right)|_{q=0}
\]
When $k=N$ and $q=0$, $\mathbb{P}(B<N)=1-2^{-N}$, $\mathbb{P}(B=N)=2^{-N},$
$\mathbb{P}(C>N)=0,$ $\mathbb{P}(C=N)=1.$ So we collect the term
$-N((1-2^{-N})+\frac{1}{2}\cdot2^{-N}-2^{-N}\cdot\frac{1}{2})=-N(1-2^{-N})$.

When $k=N-1$ and $q=0$, $\mathbb{P}(B<N-1)=1-2^{-N}-N2^{-N}$, $\mathbb{P}(B=N-1)=N2^{-N}$,
$\mathbb{P}(C>N-1)=1,$ $\mathbb{P}(C=N)=0.$ So we collect 
\[
N(1-2^{-N}-N2^{-N}+\frac{1}{2}N2^{-N}-2^{-N})=N(1-2^{-N}[2+\frac{N}{2}]).
\]

Next, we consider terms of the form 
\[
\mathbb{P}(A=k)|_{q=0}\cdot\frac{d}{dq}(\mathbb{P}(B<k)+\frac{1}{2}\mathbb{P}(B=k)-2^{-N}(\mathbb{P}(C>k)+\frac{1}{2}\mathbb{P}(C=k)))|_{q=0}.
\]

Note that $\mathbb{P}(A=k)|_{q=0}=0$ for all $k<N$. The derivative
of $\mathbb{P}(C>k)$ is $\frac{d}{dq}(1-\mathbb{P}[C\le k])=-N\mathbb{P}[\text{Bin}(N-1,1-q)=k]$.
Evaluated at $q=0$, this is 0 except when $k=N-1$, but in that case
we have $\mathbb{P}(A=N-1)=0$ when $q=0.$

The derivative of $\mathbb{P}(C=k)$ evaluated at 0 is $-N$ for $k=N,$
$N$ for $k=N-1,$ 0 otherwise. But $\mathbb{P}(A=N-1)=0$ if $q=0,$
so we collect $1\cdot(-2^{-N})\frac{1}{2}(-N).$

Collecting the terms we have obtained, and adding up, we have that:
\begin{align*}
 & -N(1-2^{-N})+N(1-2^{-N}[2+\frac{N}{2}])+(-2^{-N})\frac{1}{2}(-N)\\
= & N\left[-1+2^{-N}+1-2^{-N}[2+\frac{N}{2}]+\frac{2^{-N}}{2}\right]\\
= & N\left[2^{-N}-2^{-N+1}-(N-1)2^{-N-1}\right]=-N(N+1)2^{-(N+1)}.
\end{align*}

Overall, then, using the formula for the principal's payoff from Lemma
\ref{lem:generalredherringformulapayoffs}, the derivative of payoffs
evaluated at $q=0$ is 
\[
hN-(1-h)N(N+1)2^{-(N+1)},
\]
the sign of which equals the sign of $h/(1-h)-(N+1)2^{-(N+1)}$.
\end{proof}

\subsection{Proof of Proposition \ref{prop:finite_K}}
\begin{proof}
Write $U_{K}(q)$ for the principal's expected utility from noise
level $q$ with $K$ covariates, $N$ observations, and $h$ fraction
of hackers. Write $U(q)$ for the principal's expected utility in
the model with the same parameters from Section \ref{sec:generalredherring},
but a continuum of covariates $A=[0,1].$ From Proposition \ref{prop:generalredherring},
$U'(0)>0,$ therefore there exists some $q'>0$ so that $U(q')>U(0).$

We argue that $U_{K}(q')>U(q')$ for every finite $K\ge2.$ Note that
a maven has the same probability of proposing the true cause when
$A=[0,1]$ and when $K$ is any finite number. This is because the
maven's inference problem is restricted to only $X^{a^{*}}$ and $X^{a^{\redherr}}$
and the presence of the other covariates does not matter. For the
hacker's problem, note that the optimal behavior of the hacker is
to propose the $a$ that maximizes the number of observations where
$\hat{X}^{a}$ matches the outcome variable $Y$ in the noisy dataset.
For a hacker who has no private information about $a^{*}$, such a
covariate has the highest probability of being the true cause and
the highest probability of passing the test. The principal's utility
conditional on the hacker passing the test when $A=[0,1]$ is -1,
but this conditional utility is strictly larger than -1 when $K$
is finite as the hacker has a positive probability of choosing the
true cause. Also, the probability of the hacker passing the test with
proposal $a$ only depends on the number of observations where $\hat{X}^{a}$
matches $Y$, and the probability is an increasing function of the
number of matches. When $A=[0,1],$ the hacker can always find a covariate
that matches $Y$ in all $N$ observations in the noisy dataset, but
the hacker is sometimes unable to do so with a finite $K.$ So overall,
we must have $U_{K}(q')>U(q')>U(0)$.

Finally, we show that $U_{K}(0)-U(0)=h\left[2\frac{(1-[1-(1/2)^{N}]^{K})}{(1/2)^{N}K}-1\right]+h$,
an expression that converges to $0$ as $K\to\infty.$ Clearly, if
noise level is 0 and $A=[0,1],$ then the principal's expected utility
when facing the hacker is -1. For the case of a finite $A$, note
$X^{a^{*}}$ is perfectly correlated with $Y$. Each of the remaining
$K-1$ covariates has probability $(1/2)^{N}$ of being perfectly
correlated with $Y$, so the number of perfectly correlated variables
is $1+B$, with $B\sim\text{Binom}((1/2)^{N},K-1)$.

The hacker will recommend a perfectly correlated action at random,
so the recommendation is correct and yields of payoff of 1 with probability
$1/(1+b)$, and incorrect with probability $b/(1+b)$, for each realization
$b$ of $B$. Hence the expected payoff from facing a hacker is $\E(\frac{1-B}{1+B})$.
Using the calculation of $\E(1/(1+B))$ in \citet{chao1972negative},
\[
\E(\frac{1-B}{1+B})=2\E(\frac{1}{1+B})-1=\frac{2(1-(1-p)^{K})}{pK}-1,
\]
where $p=(1/2)^{N}$.

Combining the fact that $\lim_{K\to\infty}(U_{K}(0)-U(0))=0$ with
$U_{K}(q')>U(q')>U(0)$, there exists some $\underline{K}$ so that
$U_{K}(q')>U(q')>U_{K}(0)$ for every $K\ge\underline{K}$.
\end{proof}

\subsection{Proof of Proposition \ref{prop:no_true_cause}}
\begin{proof}
First, there exists some $\bar{q}_{1}>0$ so that for any noise level
$0\le q\le\bar{q}_{1}$, the hacker finds it optimal to report a covariate
$a$ that satisfies $X_{n}^{a}=Y_{n}$ for every observation $n$
in the data. Such a covariate has probability 0 of being correct but
the highest probability of being implemented out of all covariates
$a\in[0,1].$ If the hacker instead reports $\varnothing$, the expected
payoff is $1-\beta$. When $q=0$, the expected payoff from reporting
$a$ is $1-w_{\text{maven}}>1-\beta$ since $w_{\text{maven}}<\beta.$ The chance of such a covariate
passing the policymaker's test is continuous in noise level, so there
is some $\bar{q}_{1}>0$ so that for every noise level $0\le q\le\bar{q}_{1}$,
the hacker's optimal behavior involves reporting a covariate that
perfectly matches the outcome in the noisy data.

This means for $0\le q\le\bar{q}_{1},$ the principal's expected payoff
with dissemination noise $q$ when facing a hacker is $-V_{hacker}(q)=-(1-q)^{N}$,
with $-V_{hacker}'(q)=N(1-q)^{N-1}$.

For any $0\le q\le1/2,$ after the maven observes the two covariates
$a_{1},a_{2}\in[0,1]$ (one of them being $a^{*}$ and the other being
$a^{\redherr}$, and there is some $\beta$ probability that $a^{*}$ is
the true cause), it is optimal to either report the covariate $a\in\{a_{1},a_{2}\}$
that satisfies $X_{n}^{a}=Y_{n}$ for a larger number of observations
$n,$ or to report $\varnothing$. To see that it is suboptimal to
report any other covariate, note the maven knows that the correct
report is either $a_{1},a_{2},$ or $\varnothing$, and assigns some
posterior belief to each. At least one of the three option must have
a posterior belief that is at least 1/3, therefore the best option
out of $a_{1},a_{2},$ or $\varnothing$ must give an expected payoff
of at least $\frac{1}{3}w_{\text{maven}}$. On the other hand, reporting a covariate
$a\in[0,1]\backslash\{a_{1},a_{2}\}$ gives at most an expected payoff
of $1-w_{\text{maven}}.$ We have $\frac{1}{3}w_{\text{maven}}>1-w_{\text{maven}}$ by the hypothesis $w_{\text{maven}}>\frac{3}{4}$.

We show that there is some $\bar{q}_{2}>0$ so that for any noise
level $0\le q\le\bar{q}_{2}$, if in the noisy dataset we have (i)
$X_{n}^{a_{1}}=Y_{n}$ for all $n$, $X_{n}^{a_{2}}=1-Y_{n}$ for
all $n$, or (ii) $X_{n}^{a_{1}}=Y_{n}$ for all $n$, $X_{n}^{a_{2}}=1-Y_{n}$
for all except one $n$; or (iii) $X_{n}^{a_{1}}=Y_{n}$ for all except
one $n$, $X_{n}^{a_{2}}=1-Y_{n}$ for all $n$, then the maven reports
$a_{1}$. It suffices to show that for small enough $q$, in all three
cases the posterior probability of $a_{1}$ being the true cause exceeds
1/2 (so that the expected utility from reporting $a_{1}$ exceeds
that of reporting $\varnothing$). In case (i), this posterior probability
is \[
\frac{0.5\beta(1-q)^{2N}}{0.5\beta(1-q)^{2N}+0.5\beta q^{2N}+(1-\beta)(1-q)^{N}q^{N}},
\]
which converges to 1 as $q\to0.$ In case (ii), this posterior probability
is 
\[\frac{0.5\beta(1-q)^{2N-1}q}{0.5\beta(1-q)^{2N-1}q+0.5\beta q^{2N-1}(1-q)+(1-\beta)q^{N+1}(1-q)^{N-1}}. \]
Factoring out $q$ from the numerator and the denominator, this converges
to 1 as $q\to0.$ In case (iii), this posterior probability is \[ 
\frac{0.5\beta(1-q)^{2N-1}q}{0.5\beta(1-q)^{2N-1}q+0.5\beta q{}^{2N-1}(1-q)+(1-\beta)q^{N-1}(1-q)^{N+1}}.
\]
Factoring out $q$ from the numerator and the denominator, this converges
to 1 as $q\to0.$

The principal's expected payoff from facing the maven is the probability
that a true cause exists in the data and the maven reports $a^{*}$.
This is because the maven either reports $\varnothing$ (so the principal
gets 0), or reports a covariate that is either the true cause or gets
rejected by the policymaker. When $q=0$, the principal's expected
payoff is $\beta.$ A lower bound on the principal's payoff for $0\le q\le\bar{q}_{2}$
is $L(q):=\beta\cdot\mathbb{P}[\text{noise level \ensuremath{q} flips 0 or 1 entries in }X_{n}^{a^{*}},X_{n}^{a^{\redherr}},1\le n\le N]$.
If $a^{*}$ is the true cause and the noise flips no more than 1 entry
in $X_{n}^{a^{*}},X_{n}^{a^{\redherr}}$, then the maven sees one of cases
(i), (ii), or (iii) in the noisy data, and by the argument before
the maven will report $a^{*}$ if $q\le\bar{q}_{2}.$ Note this lower
bound is equal to the principal's expected payoff when $q=0$.

We have 
\[
L(q)=\beta\cdot(1-q)^{2N}+2N\cdot(1-q)^{2N-1}\cdot q.
\]
 The derivative is: 
\[
L'(q)=\beta\cdot[-2N(1-q)^{2N-1}+2N\cdot(1-q)^{2N-1}-2N\cdot(2N-1)\cdot(1-q)^{2N-2}\cdot q]
\]
so $L'(0)=0.$ We have that $L(q)-V_{hacker}(q)$ is a lower bound
on the principal's expected payoff with dissemination noise $q$ for
all $0\le q\le\min(\bar{q}_{1},\bar{q}_{2}),$ and the bound is equal
to the expected payoff when $q=0.$ We have $L'(0)-V_{hacker}'(0)>0,$
therefore there exists some $0<\bar{q}<\min(\bar{q}_{1},\bar{q}_{2})$
so that the lower bound on payoff $L(q)-V_{hacker}(q)$ is strictly
increasing up to $\bar{q}$. This shows any noise level $0<q<\bar{q}$
is strictly better than zero noise for the principal. 
\end{proof}

\subsection{Proof of Proposition \ref{prop:dynamic}}
\begin{proof}
Define 1 minus the state, $f=1-b.$ Define $u(q,f)$ as the principal's
expected utility today from releasing testing set with noise level
$q$ when the hacker's best guess has $1-f$ chance of being a bait
in the raw dataset. We are studying the Bellman equation 
\[
v(f)=\max\{u(q,f)+\da v\left(\frac{f(1-q)}{f(1-q)+(1-f)q}\right):q\in[0,1/2]\}
\]

First we argue that $v:[0,1]\to\mathbb{R}$ is monotone decreasing
and convex. Let $C_{B}([0,1])$ denote the set of continuous bounded
functions on $[0,1]$. Recall that $v$ is the unique fixed point
of the Bellman operator $T:C_{B}([0,1])\to C_{B}([0,1])$, with 
\[
Tw(f)=\max\{u(q;b)+\da w\left(\frac{bq}{(1-b)(1-q)+bq}\right):q\in[0,1/2]\}.
\]
Observe that $b\mapsto\frac{bq}{(1-b)(1-q)+bq}$ is concave when $q\leq1/2$
(its second derivative is $\frac{q(1-q)(2q-1)}{[(1-b)(1-q)+bq]^{3}}$).
Then when $w$ is convex and monotone decreasing, so is $b\mapsto w\left(\frac{bq}{(1-b)(1-q)+bq}\right)$,
as the composition of a concave function and a monotone decreasing
convex function is convex. Finally, $Tw$ is convex because $f\mapsto u(q,f)$
is convex (linear), and $Tw$ thus is the pointwise maximum of convex
functions. So $Tw$ is monotone decreasing. The fixed point $v$
of $T$ is the limit of $T^{n}w$, starting from any monotone decreasing
and convex $w\in C_{B}([0,1])$, so $v$ is monotone decreasing and
convex.

Observe that $f\leq\frac{f(1-q)}{f(1-q)+(1-f)q}=\ta(q,f)$, so along
any path $(q_{t},f_{t})$, $f_{t}$ is monotone (weakly) increasing.
In consequence, if $f_{t}$ is large enough, $f_{t'}$ will be large
enough for all $t'\geq t$.

Recall that 
\[
u(q,f)=(1-h)[1-q]-h[f(1-q)+(1-f)q],
\]
so 
\[
\partial_{q}u(q,f)=-1+2hf<0
\]
as $h<1/2$. Hence, we have that 
\[
u(0,f)=1-h-hf\geq u(q,f)\geq -0.5(1-h) -h(1/2)=u(1/2,f).
\]
Note that $\ta(1/2,f)=f$, so that 
\[
v(f)\geq\frac{1-h-1/2}{1-\da}.
\]

We proceed to show that $q_{t}<1/2$. Observe that if, for some $f_{t}$
it is optimal to set $q_{t}=1/2$ then $f_{t+1}=\ta(q_{t},f_{t})=f_{t}$,
and it will remain optimal to set $q_{t+1}=1/2$. This means that,
if it is optimal to set $q=1/2$ for $f$, then $v(f)=\frac{1-h-1/2}{1-\da}.$
Since $h<1/2$, $u$ is strictly decreasing in $q$. So there is a
gain in decreasing $q$ from $1/2$, which will result in transitioning
to $f'=\ta(q,f)>f=\ta(1/2,f)$. But recall that $\frac{1-h-1/2}{1-\da}$
is a lower bound on $v$. So $v(f')\geq v(f)$.
Hence, 
\[
\begin{split}u(q,f)+\da v(f')-[u(1/2,f)+\da v(f)]=(2hf-1)(q'-(1/2))+\da(v(f')-v(f))\\
\geq(2hf-1)(q'-(1/2))>0.
\end{split}
\]

Now we show that for $f$ large enough, but bounded away from 1, it
is optimal to set $q=0$. Given that $v$ is convex, it has a subdifferential:
for any $f$ there exists $\partial v(f)\in\mathbb{R}$ with the property
that $v(f')\geq v(f)+\partial v(f)(f'-f)$ for all $f'$. Since $v$
is monotone decreasing, $\partial v(f)\leq0$. Moreover, we can choose
a subdifferential for each $f$ so that $f\mapsto\partial v(f)$ is
monotone (weakly) increasing.

Let $q'<q$. Suppose that $q$ results in $f'=\ta(q,f)$ and $q'$
in $f''=\ta(q',f)$. The function $\ta$ is twice differentiable,
with derivatives 
\[
\partial_{x}\ta(x,f)=\frac{-f(1-f)}{[f(1-x)+x(1-f)]^{2}}\text{ and }\partial_{x}^{2}\ta(x,f)=\frac{2f(1-f)(1-2f)}{[f(1-x)+x(1-f)]^{2}}.
\]
Hence, $q\mapsto\ta(q,f)$ is concave when $f\geq1/2$.

Now we have: 
\begin{align*}
u(q',f)+\da v(f'')-[u(q,f)+\da v(f')] & =(2hf-1)(q'-q)+\da(v(f'')-v(f'))\\
 & \geq(2hf-1)(q'-q)+\da\partial v(f')(f''-f')\\
 & \geq(2hf-1)(q'-q)+\da\partial v(f')\partial_{q}\ta(q,f)(q'-q)\\
 & >\left[(1-2h)+\da\underbrace{\partial v(f')\frac{f(1-f)}{[f(1-q)+(1-f)q]^{2}}}_{A}\right](q-q'),\\
\end{align*}
where the first inequality uses the definition of subdifferential,
and the second the concavity of $\ta$, so that $f''-f'\leq\partial\ta(q,f)(q'-q'')$,
and the fact that $\partial v(f')\leq0$. The last inequality uses
that $f<1$. Recall that $1-2h>0$.

For $f$ close enough to $1$, and since $\partial v(f')\leq0$ are
monotone increasing and therefore bounded below, we can make $A$
as close to zero as desired. Thus, for $f<1$ close to $1$, we have
that $u(q',f)+\da v(f'')-u(q,f)+\da v(f')>0$ when $q'<q$. Hence
the solution will be to set $q=0$.

To finish the proof we show that $f_{t}\uparrow1$ and hence there
is $t^{*}$ at which $f_{t}$ is large enough that it is optimal to
set $q_{t}=0$.

Suppose that $f_{t}\uparrow f^{*}<1$. Note that if $f'=\ta(q,f)$
then $q=\frac{f(1-f')}{f(1-f')+(1-f)f'}.$ Thus (using $K$ for the
terms that do not depend on $q$ or $f$) 
\[
u(q_{t},f_{t})=K-hf_{t}-(1-2hf_{t})\left[\frac{f_{t}(1-f_{t+1})}{f_{t}(1-f_{t+1})+(1-f_{t})f_{t+1}}.\right]\rightarrow K-hf^{*}-(1-2hf^{*})\frac{1}{2}=K-\frac{1}{2}.
\]
Then for any $\ep$ there is $t$ such that $v(f_{t})=\sum_{t'\geq t}\da^{t'-t}u(q_{t'},f_{r'})<\frac{K-\frac{1}{2}}{1-\da}+\ep$.

On the other hand, if the principal sets $q_{t}=0$ it gets $u(0,f_{t})=K-hf_{t}$,
and transitions to $1=\ta(0,f_{t})$. Hence the value of setting $q=0$
at $t$ is 
\[
u(0,f_{t})+\da\frac{u(0,1)}{1-\da}=K-hf_{t}+\da\frac{K-h}{1-\da}>\frac{K-h}{1-\da}>\frac{K-1/2}{1-\da}.
\]
as $f_{t}\leq f^{*}<1$ and $h<1/2$.

Now choose $\ep$ such that $\frac{K-\frac{1}{2}}{1-\da}+\ep<\frac{K-h}{1-\da}$.
Then for $t$ large enough we have 
\[
v(f_{t})<u(0,f_{t})+\da\frac{u(0,1)}{1-\da},
\]
a contradiction because setting $q_{t}=0$ gives the principal a higher
payoff than in the optimal path.
\end{proof}

\end{document}